\documentclass[sigconf]{acmart}
\usepackage{microtype}
\usepackage{graphicx}
\usepackage{multirow}
\usepackage{multicol}
\usepackage{subcaption}
\usepackage{graphicx}
\usepackage{booktabs} 
\usepackage{bm}
\usepackage{hyperref}

\usepackage{amsmath}
\usepackage{mathtools}
\usepackage{amsthm}

\newtheorem{theorem}{Theorem}[section]
\newtheorem{proposition}[theorem]{Proposition}
\newtheorem{lemma}[theorem]{Lemma}
\newtheorem{corollary}[theorem]{Corollary}
\newtheorem{definition}[theorem]{Definition}
\newtheorem{assumption}[theorem]{Assumption}

\AtBeginDocument{%
  }

\setcopyright{acmlicensed}
\copyrightyear{2018}
\acmYear{2018}
\acmDOI{XXXXXXX.XXXXXXX}
\acmConference[Conference acronym 'XX]{Make sure to enter the correct
  conference title from your rights confirmation email}{June 03--05,
  2018}{Woodstock, NY}
\acmISBN{978-1-4503-XXXX-X/2018/06}

\begin{document}

\title{FlashEvaluator: Expanding Search Space with Parallel Sequence-Level Evaluation}

\author{Chao Feng}
\authornote{These authors contributed equally to this research.}
\correspondingauthor
\affiliation{%
  \institution{Kuaishou Technology}
  \city{Beijing}
  \country{China}
}
\email{fengchao08@kuaishou.com}

\author{Yuanhao Pu}
\authornotemark[1]
\affiliation{%
  \institution{Kuaishou Technology}
  \city{Beijing}
  \country{China}}
\email{puyuanhao@mail.ustc.edu.cn}

\author{Chenghao Zhang}
\authornotemark[1]
\affiliation{%
  \institution{Kuaishou Technology}
  \city{Beijing}
  \country{China}
}
\email{zhangchenghao03@kuaishou.com}

\author{Shanqi Liu}
\affiliation{%
 \institution{Kuaishou Technology}
 \city{Beijing}
 \country{China}}

\author{Shuchang Liu}
\affiliation{%
  \institution{Kuaishou Technology}
  \city{Beijing}
  \country{China}}

\author{Xiang Li}
\affiliation{%
  \institution{Kuaishou Technology}
  \city{Beijing}
  \country{China}}

\author{Chunjie Chen}
\affiliation{%
  \institution{Kuaishou Technology}
  \city{Beijing}
  \country{China}}

\author{Kaiqiao Zhan}
\affiliation{%
  \institution{Kuaishou Technology}
  \city{Beijing}
  \country{China}}


\renewcommand{\shortauthors}{Trovato et al.}

\begin{abstract}
The Generator-Evaluator (G-E) framework generates $K$ candidate sequences and uses an evaluator to select the highest-scoring one, which is widely used in recommender systems (RecSys) and natural language processing (NLP). Existing evaluators commonly score candidates independently. Although such evaluations can be batched, independent scoring neither models interactions among candidates nor eliminates repeated computation of request-level context and recurring candidate elements, causing the total evaluation work to grow approximately linearly with $K$. To handle with, we propose \textbf{FlashEvaluator}, a joint evaluator that scores all candidate sequences in a single forward pass. FlashEvaluator factorizes evaluation into shared request-level encoding, reusable candidate-side computation, sequence assembly by indexing, and cross-sequence interaction
for setwise comparison. We call this request-local reuse scheme \textbf{QKV-Cache}: inspired by autoregressive KV caching, it reuses context-side key/value representations across candidate sequences and, when candidate elements recur, reuses their request-conditioned representations on the query side. In repeated-item settings, the dominant item-encoding cost therefore depends on the number of distinct items rather than their total occurrences across sequences, reducing the marginal cost of evaluating additional candidates. We provide a computational analysis and evaluate FlashEvaluator on recommendation and text summarization. The results show lower latency and higher throughput with competitive recommendation and summarization quality. In an online deployment at Kuaishou with $K=50$, FlashEvaluator reduces inference latency by $44\%$ and increases QPS by $114\%$ relative to the production baseline, while yielding statistically significant gains in retention, engagement, and ecosystem metrics.
\end{abstract}

\ccsdesc[500]{Information systems~Recommender systems}
\ccsdesc[500]{Information systems~Learning to rank}
\ccsdesc[100]{Computing methodologies~Natural language processing}


\received{20 February 2007}
\received[revised]{12 March 2009}
\received[accepted]{5 June 2009}

\maketitle

\section{Introduction}
\label{sec:intro}

The Generator-Evaluator (G-E) paradigm is a common two-stage pattern for candidate generation and selection in recommender systems (RecSys), information retrieval (IR), and natural language processing (NLP). Given a context $x$, a generator produces $K$ candidate sequences, and an evaluator estimates their task-specific utilities and selects the final output. In cascaded RecSys and IR pipelines~\cite{liu2024recflow}, upstream retrieval and ranking stages reduce a large item or document corpus to a candidate pool of $M$ elements. One or more generators~\cite{yang2025comprehensive} then construct $K$ ordered lists from this pool, after which an evaluator selects one list for serving. Variants of this best-of-$K$ pattern also appear in NLP applications, where beam search, stochastic sampling, or other decoding strategies produce multiple hypotheses that are subsequently reranked by a learned evaluator, a reward model, or a task-specific metric, as in machine translation~\cite{fernandes2022quality}, summarization~\cite{romain2018rl4summarization}, retrieval-augmented generation~\cite{kelvin2020ragpretrain}, and reasoning~\cite{luyu2023pal}.

The generator and evaluator play complementary roles in this framework. The generator determines candidate coverage: for a fixed candidate set, the utility of its best output defines the oracle performance attainable by any evaluator selecting from that set. Candidate quality and diversity therefore determine whether a high-utility solution is available~\cite{yang2025comprehensive}. The evaluator, in turn, determines selection quality by estimating the relative utilities of the $K$ candidates and identifying the best one. When evaluator scores are further used as rewards or supervision during generator optimization, evaluator errors may also propagate to the
generator~\cite{zhou2025onerec}.

Most existing evaluators adopt candidate-factorized scoring:
\begin{equation*}
    s_k = f_\theta(x, L_k), \quad k \in [K],
\end{equation*}
where the score of $L_k$ does not explicitly condition on the other candidate sequences. This design is simple and can be efficiently batched, but it cannot directly exploit set-relative evidence such as pairwise dominance, comparative calibration, or contrasts among competing candidates. Such evidence can be particularly useful when multiple candidates are individually plausible but differ only in subtle aspects of task-specific utility. Recent work on setwise reranking~\cite{zhuang2024setwise} suggests that explicitly modeling inter-candidate relationships can improve relative ranking decisions.

Besides, independent evaluation may also introduce redundant computation. Although the $K$ candidates can be evaluated in parallel as a batch, batching alone does not eliminate repeated work. In RecSys and IR, a naive implementation may repeatedly encode the same request-level context and the same items when they occur in multiple candidate lists. In NLP, source- or prompt-side representations may similarly be replicated across candidate-scoring inputs. For a fixed candidate length, the total encoding work of such an implementation grows approximately linearly with $K$. Its wall-clock latency need not grow linearly because it also depends on batch size, available parallelism, memory capacity, and hardware utilization; nevertheless, the duplicated computation can reduce throughput and increase latency under constrained serving resources.

To address these limitations, we propose \textbf{FlashEvaluator}, a joint multi-sequence evaluator that produces scores for all $K$ candidate sequences in a single forward pass. FlashEvaluator decomposes evaluation into four stages: (1) request-level context encoding; (2) sequence-agnostic modeling of distinct candidate elements; (3) candidate-sequence construction through indexing, followed by sequence-specific positional encoding and intra-sequence aggregation; and (4) cross-sequence interaction for setwise comparison. This factorization separates computations that can be shared across candidates from those that must remain sequence-specific.

In the cross-attention module, context-side key and value
representations are computed once per request and reused across all candidate sequences. When the same candidate element occurs in multiple sequences, its request-conditioned query-side representation is also computed once and reused before sequence assembly. We refer to this candidate-side reuse as \textbf{Q-Cache}, and to the combination of context-side K/V reuse and candidate-side Q reuse as \textbf{QKV-Cache}. Unlike the standard autoregressive KV-cache, which reuses historical token states across decoding steps, our cache is request-local and amortizes computation across multiple candidate sequences within the same evaluation request. 

Let $U=\left|\bigcup_{k=1}^{K} L_k\right|$ denote the number of distinct candidate elements appearing across the $K$ sequences, each of length $l$. FlashEvaluator reduces the number of expensive candidate-element encodings from $Kl$ occurrences to $U$
distinct elements, while reducing request-level context encoding from $K$ executions to one. In addition to this computation reuse, the cross-sequence interaction module allows each prediction to condition on the full candidate set and therefore supports direct comparative evaluation.

The main contributions of this work are summarized as follows:
\begin{itemize}
    \item We formulate multi-sequence evaluation as a set-conditioned selection problem and identify two limitations of conventional independent evaluators: the absence of explicit cross-candidate interaction and repeated request- or candidate-side computation in naive per-sequence implementations.

    \item We propose \textbf{FlashEvaluator}, which combines request-level context reuse, sequence-agnostic candidate-element modeling, index-based sequence assembly, and cross-sequence interaction. The resulting architecture outputs all $K$ candidate scores in a single forward pass and supports both shared computation and candidate-relative comparison.

    \item We introduce a request-local \textbf{QKV-Cache} formulation that reuses context-side key/value representations and, when candidate elements recur across sequences, their request-conditioned query-side representations. We provide a computational analysis that separates these reusable encoding costs from indexing, sequence-level modeling, and cross-sequence interaction overhead.

    \item We evaluate FlashEvaluator on recommendation and text summarization benchmarks. The results show competitive recommendation and summarization quality, and substantial inference-efficiency gains. In an online deployment at Kuaishou with $K=50$, FlashEvaluator reduces inference latency by $44\%$ and increases QPS by $114\%$ relative to the production baseline, while producing statistically significant gains in retention, engagement, and ecosystem metrics.
\end{itemize}

\section{Related Work}
\label{sec:relwork}

We review three most relevant lines to FlashEvaluator: generative reranking and multi-list evaluation in RecSys/IR, multi-hypothesis generation and selection in NLP, and computation reuse for efficient multi-candidate inference. We focus on the closest methods in this section and defer broader discussions to Appendix~\ref{appdix:extended_related_work}.

\textbf{Generative reranking and multi-list evaluation.}
Classical reranking methods refine a list produced by upstream retrieval or ranking stages by modeling contextual dependencies among
items and between the request and the list ~\cite{ai2018learning,pang2020setrank,pei2019personalized,DBLP:conf/www/LiZLSCZTXH22,DBLP:conf/sigir/XiLZZDTZZY22,chi2022extr}. Generative reranking instead formulates slate construction as a structured sequence-prediction problem. Autoregressive approaches construct a list sequentially~\cite{bello2018seq2slate,DBLP:journals/corr/abs-2104-00860,DBLP:journals/corr/abs-2005-12206,DBLP:conf/kdd/ZhuLCMSZZX0025,DBLP:conf/kdd/00060HSMZ0G23}, whereas non-autoregressive approaches predict or refine multiple positions jointly to improve generation efficiency~\cite{DBLP:conf/iclr/JiangGQMR19,DBLP:conf/www/Liu0GPZ21,meng2025generative,ren2024nar4rec,DBLP:conf/www/WangWKWCTZXW25,zhang2026dual,lin2024dcdr}. These methods differ in whether they produce a single output list or a collection of candidate lists, and not all of them contain a separate evaluator.

Explicit G-E pipelines use one or more generators to propose candidate lists and then employ a list-level evaluator to estimate their utilities and select the final output~\cite{shi2023pier,DBLP:journals/corr/abs-2102-12057,DBLP:conf/atal/ZhaoLFGZHCPD24,DBLP:conf/www/XuXCLSLWZZ26,yang2025comprehensive}. Setwise ranking methods also compare multiple documents jointly~\cite{zhuang2024setwise}, but they typically rank atomic document candidates rather than complete generated lists. Generative retrieval provides a related but distinct formulation: items are represented by sequences of Semantic IDs, and retrieval is performed by generating valid ID paths~\cite{DBLP:conf/nips/RajputMSKVHHT0S23,zhou2025onerec}. Such methods can be interpreted through a G-E lens when multiple generated paths are explicitly rescored by a separate reward or evaluation model.

\textbf{Multi-hypothesis generation and selection.} In NLP, early discriminative and minimum-error-rate methods optimized decisions over $N$-best hypothesis sets~\cite{collins2002perceptron,och2003mert}, while recent quality-aware decoding methods combine neural generation with reference-free quality-estimation reranking or reference-based selection~\cite{fernandes2022quality}. Minimum Bayes risk (MBR) decoding and consensus-based methods are explicitly set-dependent: they evaluate each candidate according to its expected utility or agreement with other hypotheses~\cite{kumar2004mbr,rosti2007consensus,hildebrand2008combination}. Self-consistency follows a related principle for reasoning by aggregating answers obtained from multiple sampled reasoning paths~\cite{wang2023selfconsistency}. These methods are distinct from two-pass or iterative refinement architectures such as Deliberation Networks~\cite{xia2017deliberation}, iterative refinement~\cite{lee2018refinenet}, and Mask-Predict~\cite{ghazvininejad2019maskpredict}. The latter revise provisional outputs over successive decoding passes, rather than evaluating a fixed set of $K$ complete candidates and selecting one of them. Moreover, although MBR and consensus decoding exploit candidate-set information at the decision-rule level, a common learned-reranker design still computes a separate neural score for each hypothesis. FlashEvaluator instead focuses on shared neural computation and explicit representation-level interaction among complete candidate sequences before producing their scores.

\textbf{Caching and computation reuse.}
Standard autoregressive KV caching stores the key and value states of previously processed tokens and reuses them at subsequent decoding steps. Existing work reduces its memory and serving overhead through multi-query attention (MQA)~\cite{DBLP:journals/corr/abs-1911-02150}, grouped-query attention (GQA)~\cite{DBLP:conf/emnlp/AinslieLJZLS23}, layer-condensed caching~\cite{DBLP:conf/acl/WuT24}, token eviction
~\cite{3666122.3667628,DBLP:conf/iclr/XiaoTCHL24}, quantization~\cite{DBLP:conf/icml/LiuYJZXBC024,DBLP:journals/corr/abs-2407-21118}, and paged memory management~\cite{DBLP:conf/sosp/KwonLZ0ZY0ZS23}. These techniques primarily target temporal reuse and memory management during autoregressive decoding, and are therefore complementary to reuse across multiple candidates within one evaluation request.

The closest reuse mechanisms in RecSys are M-FALCON in HSTU~\cite{DBLP:conf/icml/ZhaiLLWLCGGGHLS24}
and the Context Cache Module in YOLOR~\cite{wang2025yolor}. M-FALCON batches target-aware scoring for candidate items and can reuse the key/value states of the shared user history across candidate microbatches. YOLOR hierarchically extracts multi-scale contexts and reuses them across candidate permutations through a tree-structured cache. Our focus differs in both evaluation granularity and interaction structure: FlashEvaluator jointly evaluates a finite set of candidate sequences, reuses request-side key/value representations and reusable candidate-side representations before sequence assembly, and explicitly models interactions among the resulting sequence representations. We use the term \textbf{QKV-Cache} to denote this request-local reuse across candidate sequences.
\section{Preliminaries}
\label{sec:preliminaries}

\paragraph{Notations.}
We consider a general Generator-Evaluator (G-E) framework for multi-candidate sequence selection. Let $i\in\mathcal I$ denote a basic element (e.g., an item in RecSys, a document in IR, or a generated unit
in NLP), and let $x\in\mathcal X$ denote the query context. An upstream retrieval or generation process provides a candidate set $C(x)=\{i_1,i_2,\dots,i_M\}\subseteq\mathcal I$.
The generator constructs $K$ candidate sequences $L_{1:K}=\{L_1,\dots,L_K\}$ where each candidate sequence $L_k=(i_{k,1},\dots,i_{k,l})$ contains $l$ elements sampled from $C(x)$.

We define a candidate instance as $z=(x,C,L_{1:K})$ which contains the context, candidate pool, and all generated sequences. An evaluator maps the complete candidate instance into a score vector $E_\theta:\mathcal Z\rightarrow\mathbb R^K$ where $E_\theta(z)=(\hat{s}_1,\dots,\hat{s}_K)$ and $\hat{s}_k$ denotes the predicted utility score of candidate sequence $L_k$. The final output of the G--E framework is selected by
\begin{equation}
\hat{y}_\theta(z)=\arg\max_{k\in[K]}\hat{s}_k .
\end{equation}
where $[K]=\{1,2,\dots,K\}$.

\paragraph{Training sample.}
A training example is denoted as $(z,y)$, where $y\in[K]$ indicates an observed supervision label derived from task-specific feedback or a proxy sequence utility. We assume the training samples $S=\{(z^{(t)},y^{(t)})\}_{t=1}^{n}\sim \mathcal D^n$ are independently drawn from an unknown data distribution $\mathcal D$.

\begin{definition}[Ground-truth sequence utility]
There exists an unknown utility function $u^\star:\mathcal X\times\mathcal I^M\times\mathcal I^{l\times K}\times[K]\rightarrow\mathbb R$ such that for each candidate sequence $L_k$ in an instance $z$,
\begin{equation}
u_k^\star(z)=u^\star(z,k).
\end{equation}
The Bayes-optimal candidate is defined as
\begin{equation}
y^\star(z)=\arg\max_{k\in[K]}u_k^\star(z).
\end{equation}
\end{definition}

The objective of the evaluator is to minimize the Top-1 risk:
\begin{equation}
\label{eq:top1-risk}
\mathcal L_{\mathrm{Top-1}}(\theta)=\mathbb P_{z\sim\mathcal D}\left(\hat y_\theta(z)\neq y^\star(z)\right).
\end{equation}

Since directly optimizing it is generally intractable, we generally train the evaluator using a differentiable surrogate loss $\ell$. For the predicted score vector $\hat{\bm s}=E_\theta(z)$
and label $y$, the population surrogate risk is defined as
\begin{equation}
\label{eq:surrogate-risk}
\mathcal R(\theta)
=
\mathbb E_{(z,y)\sim\mathcal D}
\left[
\ell(\hat{\bm s},y)
\right].
\end{equation}
whose corresponding empirical risk on training sample $S$ is
\begin{equation}
\label{eq:empirical-risk}
\widehat{\mathcal R}_S(\theta)
=
\frac1n
\sum_{t=1}^{n}
\ell
\left(
E_\theta(z^{(t)}),
y^{(t)}
\right).
\end{equation}

\section{FlashEvaluator}
\label{sec:method}

\begin{figure*}[t]
    \centering
    \begin{subfigure}[t]{0.49\textwidth}
        \centering
        \includegraphics[width=\textwidth]{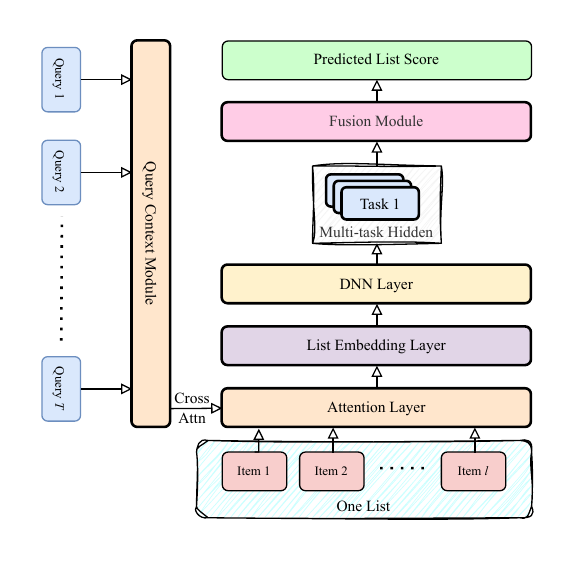}
        \caption{\textbf{Candidate-factorized evaluator.}
        Candidate sequences are scored independently, resulting in repeated
        context- and candidate-side computation.}
        \label{fig:baseline}
    \end{subfigure}
    \hfill
    \begin{subfigure}[t]{0.49\textwidth}
        \centering
        \includegraphics[width=\textwidth]{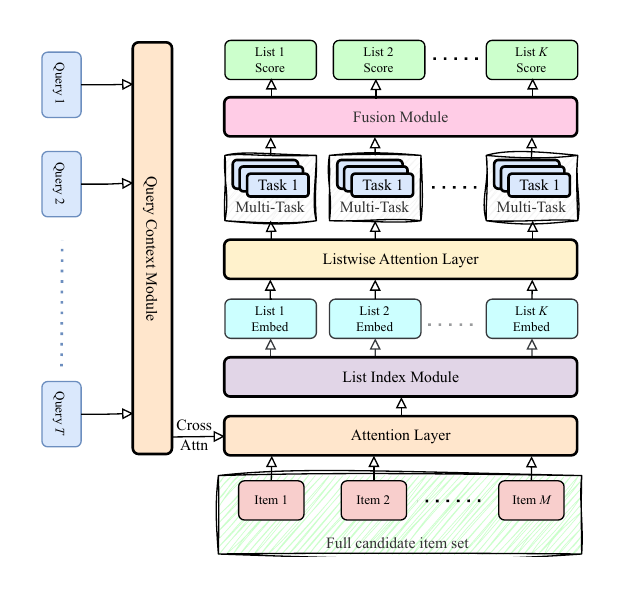}
        \caption{\textbf{FlashEvaluator.}
        Shared representations are computed once, candidate sequences are
        assembled by indexing, and their representations interact before scoring.}
        \label{fig:ours}
    \end{subfigure}
    \caption{Comparison between a conventional candidate-factorized evaluator
    and the proposed joint evaluator.}
    \label{fig:architecture}
\end{figure*}

Given a candidate instance $z$, conventional evaluators commonly adopt candidate-factorized scoring,
\[
E_{\theta}^{\mathrm{ind}}(z)=\bigl(f_{\theta}(x,L_1),\ldots,f_{\theta}(x,L_K)\bigr).
\]
As discussed above, this formulation may repeatedly construct representations of the same request context and recurring candidate elements. Moreover, the predicted score cannot directly condition on the other competing candidate sequences. FlashEvaluator instead implements the joint mapping $E_{\theta}(z)\in\mathbb{R}^{K}$. As illustrated in Figure~\ref{fig:ours}, it factorizes evaluation into four stages: shared request- and candidate-side encoding, index-based candidate-sequence assembly, intra-sequence aggregation, and cross-sequence interaction.

\subsection{Architecture Overview}
\label{sec:architecture}

Let $\mathcal{U}(z)=\bigcup_{k=1}^{K}\{i_{k,1},\ldots,i_{k,l}\}, U=|\mathcal{U}(z)|\leq M$ denote the set of distinct elements appearing in the generated candidate sequences, while we can also naively encode the entire upstream candidate pool by setting $\mathcal{U}(z)=C(x)$ and $U=M$. Operating on $\mathcal{U}(z)$ exposes the maximum amount of candidate-side reuse, since each distinct element is encoded only once regardless of how often it occurs in $L_{1:K}$. FlashEvaluator then computes the following modules:
\begin{equation*}
\bm{H}=\operatorname{Enc}_\theta\bigl(x,\mathcal{U}(z)\bigr)\rightarrow\bm{c}_{1:K}=\operatorname{Idx}_\theta\bigl(\bm{H},L_{1:K}\bigr)\rightarrow\widehat{\bm{s}}=\operatorname{Inter}_\theta\bigl(\bm{c}_{1:K}\bigr),
\end{equation*}
where $\bm{H}$ contains reusable representations of the candidate elements, $\bm{c}_{1:K}=(\bm{c}_1,\ldots,\bm{c}_K)$ are the sequence-level representations, and $\widehat{\bm{s}}=(\hat{s}_1,\ldots,\hat{s}_K)$ is the final score vector. $\operatorname{Enc}_{\theta}$ performs shared request- and candidate-side encoding, $\operatorname{Idx}_{\theta}$ performs index-based sequence assembly and intra-sequence aggregation, and $\operatorname{Inter}_{\theta}$ models cross-sequence interactions before producing the $K$ scores.

\subsection{Shared Request and Candidate Encoding}
\label{sec:shared-encoding}

Each element $i\in\mathcal{U}(z)$ is mapped to a dense representation $\bm{h}_i\in\mathbb{R}^{d}$. Then, apply self-attention over all distinct candidate elements:
\[
\bm{H}^{(1)}=\operatorname{SelfAttn}\bigl(\{\bm{h}_{i}:i\in\mathcal{U}(z)\}\bigr),
\]
which captures dependencies among elements before they are assembled into individual candidate sequences. The query context is encoded by the Query Context Module as $\bm{Q}=\{\bm{q}_1,\bm{q}_2,\ldots,\bm{q}_T\}$. We inject the request-level context into the candidate representations through cross-attention:
\[
\bm{H}^{(2)}=\operatorname{CrossAttn}\bigl(\bm{H}^{(1)},\bm{Q}\bigr).
\]
The resulting representation of each candidate element is conditioned on the shared query context and can be reused in every sequence where that element appears.

This computation naturally supports request-local representation reuse. The context-side key and value representations derived from $\bm{Q}$ are computed only once for the candidate instance, rather
than separately for each candidate sequence. We refer to the reuse of candidate-side query representations as \emph{Q-Cache}, and to its combination with context-side key/value reuse as \emph{QKV-Cache}. Unlike autoregressive KV caching, which reuses historical states across decoding steps, QKV-Cache amortizes computation across multiple complete candidate sequences within the same evaluation request.

\subsection{Index-Based Sequence Assembly}
\label{sec:index-assembly}

For each candidate sequence $L_k=(i_{k,1},\ldots,i_{k,l})$, we construct its encoded representation $\bm{H}^{(2)}_k$ by indexing the corresponding element representations from $\bm{H}^{(2)}$:
\[\bm{H}^{(2)}_k=\bigl(\bm{h}^{(2)}_{i_{k,1}},\ldots,\bm{h}^{(2)}_{i_{k,l}}\bigr).\]
Sequence-position embeddings are added after indexing. Therefore, the same shared element representation can be reused while still receiving different positional information when the element appears at different positions. We prepend a \texttt{[CLS]} token to each indexed sequence and apply a shared intra-sequence encoder $\bm{c}_k=\operatorname{SeqEnc}_\theta\bigl([\texttt{CLS};\bm{H}^{(2)}_k]\bigr)$, where $\bm{c}_k\in\mathbb{R}^{d}$ denotes the resulting sequence-level representation. The $K$ candidate sequences share the same encoder and can be processed in parallel.

\subsection{Cross-Sequence Interaction}
\label{sec:cross-sequence}

After obtaining the sequence representations $\bm{c}_{1:K}$, FlashEvaluator jointly models their interactions:
\[
\bm{C}^{\prime}=\operatorname{SelfAttn}\bigl(\{\bm{c}_1,\bm{c}_2,\ldots,\bm{c}_K\}\bigr).
\]
This cross-sequence interaction allows the representation of each candidate sequence to condition on the other competing sequences. Consequently, the evaluator can capture comparative information such as relative quality among the generated candidates. Finally, a shared prediction head maps each updated sequence representation to a scalar score:
\[
\hat{s}_k=f_{\theta}\bigl(\bm{C}^{\prime}_k\bigr),\quad E_{\theta}(z)=(\hat{s}_1,\ldots,\hat{s}_K).
\]
For multi-objective applications, the shared representation $\bm{C}^{\prime}_k$ can be connected to multiple task-specific prediction heads followed by a fusion module. This extension does not change the shared encoding, sequence assembly, or cross-sequence interaction mechanisms. 

\subsection{Learning Objectives}
\label{sec:learning-objectives}

FlashEvaluator can be trained with either pointwise or listwise objective. For pointwise objectives, we can use the mean-squared error when a scalar target $s_k^{\star}$ is available for each candidate sequence, and binary cross-entropy for binary or probabilistic targets $s_k^{\star}\in[0,1]$:
\[
\mathcal{L}_{\mathrm{MSE}}=\frac{1}{K}\sum_{k=1}^{K}\bigl(\hat{s}_k-s_k^{\star}\bigr)^2.
\]
\[
\mathcal{L}_{\mathrm{BCE}}=-\frac{1}{K}\sum_{k=1}^{K}\left[s_k^{\star}\log\sigma(\hat{s}_k)+(1-s_k^{\star})\log\bigl(1-sigma(\hat{s}_k)\bigr)\right],
\]
where $\sigma(\cdot)$ is the sigmoid function. For listwise objective, given the optimal candidate label $y\in[K]$ defined in Section~\ref{sec:preliminaries}, the evaluator can also be trained with the listwise softmax cross-entropy
\[
\mathcal{L}_{\mathrm{CE}}=-\log\frac{\exp(\hat{s}_{y})}{\sum_{k=1}^{K}\exp(\hat{s}_{k})},
\]
which directly optimizes the relative scores of all candidate sequences within the same candidate instance. The joint architecture is naturally suited to listwise training because all $K$ scores are produced in the same forward computation graph. For the recommendation experiments, the listwise label $y$ is constructed from the proxy utility $r_k$, while the pointwise variant uses item-level binary supervision, as detailed in
Appendix~\ref{appdix:datasets}.

\section{Theoretical properties}
\label{sec:theory}
We analyze the statistical properties of the two evaluator architectures introduced above. The conventional evaluator constructs its score vector by independently applying a shared scalar scoring function to each candidate sequence, whereas FlashEvaluator directly produces a jointly conditioned score vector. Notice that the evaluator architecture and the learning objective are conceptually orthogonal: either architecture can, in principle, be trained with pointwise or listwise supervision. To isolate the effect of score construction, our generalization analysis evaluates both architectures using the same listwise surrogate defined in Section~\ref{sec:preliminaries}. Appendix~\ref{appdix:robustness} provides separate objective-level characterizations of sequence-level squared regression and listwise softmax training under a decomposable additive feedback shift. Appendix~\ref{appdix:ablation} empirically compares the item-level BCE and listwise objectives
used in the RecFlow experiments.

\subsection{Generalization Bounds}
\label{sec:generalization}
We focus on the population surrogate risk $\mathcal{R}(\theta)$ and its empirical counterpart
$\widehat{\mathcal{R}}_S(\theta)$ defined in Eq.~(\ref{eq:surrogate-risk})-(\ref{eq:empirical-risk}). For independent evaluator, the corresponding risks $\mathcal{R}^{\mathrm{ind}}(\theta)$ and $\widehat{\mathcal{R}}^{\mathrm{ind}}_S(\theta)$ are defined by replacing $E_\theta(z)$ with $E_\theta^{\mathrm{ind}}(z)$.

Recall that
\[
E_\theta^{\mathrm{ind}}(z)=\bigl(f_\theta(x,L_1),\ldots,f_\theta(x,L_K)\bigr),\quad E_\theta(z)=(\hat{s}_1,\ldots,\hat{s}_K).
\]
We consider the induced function classes
\begin{align}
\mathcal{H}_{\mathrm{joint}}&=\bigl\{z\mapsto E_\theta(z):\theta\in\Theta_{\mathrm{joint}}\bigr\},\\
\mathcal{H}_{\mathrm{ind}}&=\bigl\{z\mapsto E_\theta^{\mathrm{ind}}(z):\theta\in\Theta_{\mathrm{ind}}\bigr\},
\end{align}
and define the scalar base class of the independent evaluator as
\begin{equation}
\mathcal{H}_{\mathrm{base}}=\bigl\{(x,L)\mapsto f_\theta(x,L):\theta\in\Theta_{\mathrm{ind}}\bigr\}.
\end{equation}

\begin{assumption}[Loss regularity]
\label{ass:loss}
The surrogate loss used in the analysis is normalized or clipped such that, for every $y\in[K]$, $0\leq \ell(\bm{s},y)\leq 1$, and there exists a constant $L_\ell>0$, independent of $K$, such that
\[
\bigl|\ell(\bm{s},y)-\ell(\bm{s}',y)\bigr|\leq L_\ell\|\bm{s}-\bm{s}'\|_2
\]
for all $\bm{s},\bm{s}'\in\mathbb{R}^K$.
\end{assumption}
\begin{assumption}[Capacity scaling]
\label{ass:capacity}
There exist constants $C_{\mathrm{base}}$, $C_{\mathrm{hid}}$,
and $\Lambda$, independent of the sample size $m$ and the number
of candidate sequences $K$, such that, for every finite collection
$S_m$ of $m$ valid inputs from the respective input spaces,
\[
\widehat{\mathfrak R}_{S_m}(\mathcal H_{\mathrm{base}})
\le \frac{C_{\mathrm{base}}}{\sqrt m},
\quad
\widehat{\mathfrak R}_{S_m}(\mathcal H_{\mathrm{joint}})
\le \frac{\Lambda C_{\mathrm{hid}}}{\sqrt m}.
\]
\end{assumption}

The first bound controls the scalar function repeatedly applied by the independent evaluator. The second places a global capacity constraint on the jointly produced score vector. Such conditions are consistent with standard norm-based neural-network bounds\citep{neyshabur2015norm,golowich2018size}, while a sufficient condition further explains it with Lemma~\ref{appdix:lem:joint-capacity}.

We then compare the two evaluator classes under generalization analysis. Detailed proofs are provided in Appendix~\ref{appdix:generalization}.

\begin{theorem}[Generalization bounds]
\label{thm:comparison}
Suppose Assumptions~\ref{ass:loss} and~\ref{ass:capacity} hold.
For any $\delta\in(0,1)$, with probability at least $1-\delta$, the
following bounds hold uniformly over the respective parameter spaces:
\begin{align}
\mathcal{R}(\theta)-\widehat{\mathcal{R}}_S(\theta)&\leq\frac{2\sqrt{2}L_\ell\Lambda C_{\mathrm{hid}}}{\sqrt{n}}+3\sqrt{\frac{\log(4/\delta)}{2n}},\\
\mathcal{R}^{\mathrm{ind}}(\theta)-\widehat{\mathcal{R}}^{\mathrm{ind}}_S(\theta)&\leq2\sqrt{2}L_\ell C_{\mathrm{base}}\sqrt{\frac{K}{n}}+3\sqrt{\frac{\log(4/\delta)}{2n}}.
\end{align}
\end{theorem}

Theorem~\ref{thm:comparison} shows that the leading capacity term of the joint evaluator is $\mathcal{O}(n^{-1/2})$, whereas the bound obtained by concatenating $K$ independent scores is $\mathcal{O}(\sqrt{K/n})$. In particular, if $\frac{C_{\mathrm{base}}}{\Lambda C_{\mathrm{hid}}}=\Theta(1)$, the ratio between the two leading terms scales as $\Theta(\sqrt{K})$. 

\textbf{Robustness.}
Appendix~\ref{appdix:robustness} separately characterizes the sensitivity of pointwise
squared regression and listwise softmax training under a decomposable additive feedback shift. Pointwise regression is affected by both the common candidate-set-level shift and candidate-specific distortions. In contrast, listwise softmax is exactly invariant to the common additive shift, while its local sensitivity depends only on relative candidate-specific distortions. Since these excess risks are defined under different surrogate losses, we do not compare their absolute magnitudes.
  
\subsection{Computational Complexity and Scalability}
\label{sec:complexity}

We further compare the forward-pass costs of the two
evaluator architectures. The independent evaluator repeatedly
constructs representations for all $Kl$ element occurrences, whereas
FlashEvaluator constructs a request-conditioned representation only
once for each of the $U$ distinct elements and reuses it through
indexing.
\begin{definition}[Reuse factor]
\label{def:reuse}
For a candidate instance $z$, we define the element reuse factor as
\[
\rho(z)
=
\frac{Kl}{U}.
\]
Since $U\leq Kl$, we have $\rho(z)\geq 1$. When the entire upstream candidate pool is encoded instead of $\mathcal{U}(z)$, $U$ is replaced by $M$.
\end{definition}

The following proposition characterizes the computational benefit of
reuse. Its detailed proof is provided in
Appendix~\ref{appdix:complexity}.

\begin{proposition}[Computational Advantage via Reuse]
\label{prop:complexity-advantage}
When candidate-set encoding dominates the remaining lightweight
operations, the computational-cost ratio satisfies
\[
\frac{T_{\mathrm{joint}}}{T_{\mathrm{ind}}}
\approx
\frac{U}{Kl}
+
\mathcal{O}(\epsilon)
=
\frac{1}{\rho}
+
\mathcal{O}(\epsilon),
\]
where $\epsilon$ summarizes the normalized overhead of index-based
sequence assembly, intra-sequence aggregation, score prediction, and
cross-sequence interaction.
\end{proposition}

Proposition~\ref{prop:complexity-advantage} shows that the dominant
representation-construction cost of
$E_\theta^{\mathrm{ind}}(\cdot)$ scales with the total number of
element occurrences, namely $Kl$, whereas that of
$E_\theta(\cdot)$ scales with the number of distinct elements $U$.
The resulting advantage therefore increases as the same elements are
reused across more candidate sequences.

When the number of candidate sequences increases from $K$ to $K'$,
the dominant encoding cost of the independent evaluator increases
proportionally to
\[
\Delta T_{\mathrm{ind}}
\propto
(K'-K)l.
\]
For FlashEvaluator, the encoding cost depends only on
whether the newly added sequences introduce previously unseen
elements. If $U$ remains unchanged, its marginal candidate-element
encoding cost is zero. The total additional cost is not zero but not significant, since the candidate-set interaction and sequence-dependent
operations incur only a small overhead relative to the dominant embedding-compression cost.

Therefore, FlashEvaluator's main computational advantage lies in
amortizing the dominant request- and candidate-side representation
construction across candidate sequences, allowing more candidates to
be evaluated under the same serving budget when this reusable encoding cost dominates.
\begin{table*}[t]
    \centering
    \caption{Statistics of the datasets used in the experiments.}
    \label{tab:dataset_stats}
    \begin{tabular}{lccccc}
        \toprule
        Dataset & Domain & \# Requests & \# Items & Candidate Pool Size & Target Length \\
        \midrule
        ML-1M & RecSys & 161,646 & 3,043 & 50 & 6 \\
        Amazon-Books & RecSys & 309,917 & 38,121 & 50 & 6 \\
        RecFlow & RecSys & 3,308,233 & 14,181,768 & 120 & 6 \\
        CNN/DM & NLP & 30,000 & -- & -- & 1 sentence \\
        \bottomrule
    \end{tabular}
\end{table*}
\begin{table*}[t]
\centering
\caption{Performance comparison on three recommendation datasets.
N@6, P@6, R@6, and F1@6 denote NDCG@6, Precision@6,
Recall@6, and F1@6, respectively. The best and second-best
results are highlighted in \textbf{bold} and \underline{underlined},
respectively.}
\label{tab:main_results}
\setlength{\tabcolsep}{3.2pt}
\begin{tabular}{ll*{3}{cccc}}
\toprule
\multirow{2}{*}{Category}
& \multirow{2}{*}{Model}
& \multicolumn{4}{c}{ML-1M}
& \multicolumn{4}{c}{Amazon-Books}
& \multicolumn{4}{c}{RecFlow} \\
\cmidrule(lr){3-6}
\cmidrule(lr){7-10}
\cmidrule(lr){11-14}
& & N@6 & P@6 & R@6 & F1@6
  & N@6 & P@6 & R@6 & F1@6
  & N@6 & P@6 & R@6 & F1@6 \\
\midrule

\multirow{6}{*}{\shortstack{G-Only}}
& DNN
& 0.5950 & 0.4539 & 0.5542 & 0.4876
& 0.6448 & 0.5072 & 0.6125 & 0.5472
& 0.1584 & 0.0793 & 0.2069 & 0.1084 \\

& DCN
& 0.5981 & 0.4561 & 0.5573 & 0.4901
& 0.6683 & 0.5298 & 0.6461 & 0.5701
& 0.1597 & 0.0795 & 0.2083 & 0.1088 \\

& Seq2Slate
& 0.6222 & 0.4867 & 0.5927 & 0.5225
& 0.6952 & 0.5654 & 0.6871 & 0.6078
& 0.1693 & 0.0821 & 0.2134 & 0.1130 \\

& DLCM
& 0.6061 & 0.4643 & 0.5667 & 0.4988
& 0.6597 & 0.5242 & 0.6396 & 0.5641
& 0.1747 & 0.0861 & 0.2240 & 0.1169 \\

& SetRank
& 0.7154 & 0.5720 & 0.6933 & 0.6132
& 0.8014 & 0.6635 & 0.8145 & 0.7156
& 0.1823 & 0.0896 & 0.2344 & 0.1225 \\

& PRM
& 0.7081 & 0.5639 & 0.6843 & 0.6049
& 0.7992 & 0.6603 & 0.8107 & 0.7122
& 0.1840 & 0.0905 & 0.2368 & 0.1238 \\
\midrule
\multirow{1}{*}{\shortstack{E-Only}}
& YOLOR
& 0.7203 & 0.5769 & 0.6994 & 0.6186
& 0.8088 & 0.6708 & 0.8240 & 0.7235
& 0.1895 & 0.0924 & 0.2416 & 0.1265 \\

\midrule
\multirow{4}{*}{\shortstack{G-E}}
& PIER
& 0.7146 & 0.5715 & 0.6929 & 0.6128
& 0.7987 & 0.6613 & 0.8118 & 0.7130
& \underline{0.1910} & 0.0935 & 0.2431 & 0.1277 \\

& PIER+Flash
& 0.7228 & 0.5793 & 0.7023 & 0.6212
& 0.8079 & 0.6700 & 0.8233 & 0.7228
& \textbf{0.1925} & \textbf{0.0938} & \textbf{0.2446} & \textbf{0.1284} \\

& NAR4Rec
& \underline{0.7348} & 0.5912 & 0.7162 & 0.6338
& \underline{0.8188} & 0.6807 & 0.8365 & 0.7341
& 0.1792 & 0.0880 & 0.2297 & 0.1203 \\

& NAR4Rec+Flash
& \textbf{0.7426} & \textbf{0.5987} & \textbf{0.7250} & \textbf{0.6419}
& \textbf{0.8274} & \textbf{0.6891} & \textbf{0.8471} & \textbf{0.7431}
& 0.1818 & 0.0891 & 0.2330 & 0.1218 \\

\bottomrule
\end{tabular}
\end{table*}

\begin{table*}[ht]
\centering
\caption{Performance comparison of different reranking methods on CNN/DailyMail using three generator backbones. The $K=64$ expands the candidate set under a lower latency to SimCLS $K=16$. The best and second-best results excluding Oracle are highlighted in \textbf{bold} and \underline{underlined}, respectively.}
\label{tab:summarization_rerank}
\begin{tabular}{lcccccccccccc}
\toprule
\multirow{2}{*}{\textbf{Method}} & \multicolumn{4}{c}{\textbf{T5-base}} & \multicolumn{4}{c}{\textbf{BART-Large}} & \multicolumn{4}{c}{\textbf{Llama-3.1-8B (3-shot)}}\\
\cmidrule(lr){2-5}\cmidrule(lr){6-9}\cmidrule(lr){10-13}
& R-1 & R-2 & R-L & BS & R-1 & R-2 & R-L & BS & R-1 & R-2 & R-L & BS \\
\midrule
\textit{Bounds} & & & & & & & & & & & & \\
Base (Beam 1) & 40.81 & 18.37 & 37.81 & 87.23 & 43.63 & 20.66 & 40.47 & 87.88 & 38.78 & 13.44 & 35.08 & 87.53 \\
Oracle ($K=16$) & 48.21 & 24.95 & 45.29 & 88.32 & 52.43 & 28.73 & 49.42 & 89.30 & 46.35 & 19.72 & 42.75 & 88.49\\
\midrule
\textit{Rerankers} & & & & & & & & & & & & \\
RankGPT (Qwen2.5-32B)
& 41.09 & 18.35 & 37.99 & 87.32
& \underline{44.54} & 21.17 & 41.33 & 88.04
& 39.00 & 13.82 & 35.33 & 87.51 \\
SimCLS ($K=16$)
& \underline{42.12} & \underline{19.33}
& \underline{39.14} & \underline{87.51}
& 44.46 & 21.12 & 41.29 & 88.03
& \underline{41.29} & \underline{15.53}
& \underline{37.48} & 87.60 \\
FlashEvaluator ($K=16$)
& 41.28 & 18.71 & 38.35 & 87.42
& 44.52 & \underline{21.35}
& \underline{41.48} & \underline{88.20}
& 40.46 & 14.65 & 36.70 & \underline{87.77} \\
FlashEvaluator ($K=64$) & \textbf{42.36} & \textbf{19.54} & \textbf{39.35} & \textbf{87.60} & \textbf{44.86} & \textbf{21.67} & \textbf{41.81} & \textbf{88.29} & \textbf{41.47} & \textbf{15.70} & \textbf{37.66} & \textbf{87.89} \\
\bottomrule
\end{tabular}
\end{table*}

\section{Experiments}
\label{sec:experiments}

We evaluate \textbf{FlashEvaluator} through both offline experiments and online deployment. The offline evaluation covers three recommendation datasets and one text-summarization benchmark, including public user-item interactions, industrial multi-stage recommendation logs, and natural language generation. We further deploy \textbf{FlashEvaluator} in Kuaishou's production short-video recommendation system and conduct an online A/B test under real-world latency constraints and multi-objective serving targets. Specifically, we aim to answer the following research questions:

\begin{itemize}
    \item \textbf{RQ1:} How does FlashEvaluator compare with representative reranking methods in recommendation effectiveness?

    \item \textbf{RQ2:} Can FlashEvaluator generalize beyond recommendation to multi-candidate selection in text summarization?

    \item \textbf{RQ3:} How does FlashEvaluator scale with $K$ in terms of inference latency and throughput?

    \item \textbf{RQ4:} How does FlashEvaluator perform when deployed in a real-world multi-objective recommendation system?
\end{itemize}

\subsection{Experimental Setup}
\label{sec:experiments_setup}

\subsubsection{Datasets}

We conduct recommendation experiments on two public interaction datasets, \textbf{ML-1M}~\cite{harper2016movielens} and \textbf{Amazon-Books}~\cite{he2016ups}, and one industrial multi-stage dataset, \textbf{RecFlow}~\cite{liu2024recflow}. We further evaluate the cross-domain applicability of FlashEvaluator on the \textbf{CNN/DailyMail}~\cite{hermann2015teaching} text-summarization benchmark. Dataset statistics are summarized in Table~\ref{tab:dataset_stats}.

Since \textbf{ML-1M} and \textbf{Amazon-Books} do not provide request-level candidate pools from an upstream ranking stage, we train \textbf{BPR-MF}~\citep{rendle2009bpr} to simulate first-stage retrieval. For each user, the last six interactions are reserved as the test target and excluded from all training procedures, while the preceding interactions are used to construct length-six reranking samples. Training candidates are retrieved from the user's top-200 BPR-MF results over the full item corpus, and each test candidate pool contains 50 distinct items. \textbf{RecFlow} directly provides request-level candidates and features from a real multi-stage recommendation pipeline; following its official protocol, we use 120 upstream candidates per request, and a target sequence length of six. For \textbf{CNN/DailyMail}, we generate $K=16$ summaries for each reranker. We additionally report an expanded-candidate setting in which FlashEvaluator selects from $K=64$ summaries, examining whether its efficiency advantage can be converted into improved quality.

We additionally deploy FlashEvaluator in the reranking stage of the single-column recommendation feed on \textbf{Kuaishou}'s main app, which serves over 400 million daily active users with an average daily time spent exceeding two hours per user, providing a large-scale and latency-sensitive production environment. The production baseline and FlashEvaluator are assigned to two disjoint 10\% traffic buckets for seven days, using the same candidate generators, input features, and multi-objective serving targets.

\begin{figure*}[!ht]
    \centering
    \begin{subfigure}[t]{0.3\textwidth}
        \centering
        \includegraphics[width=\textwidth]{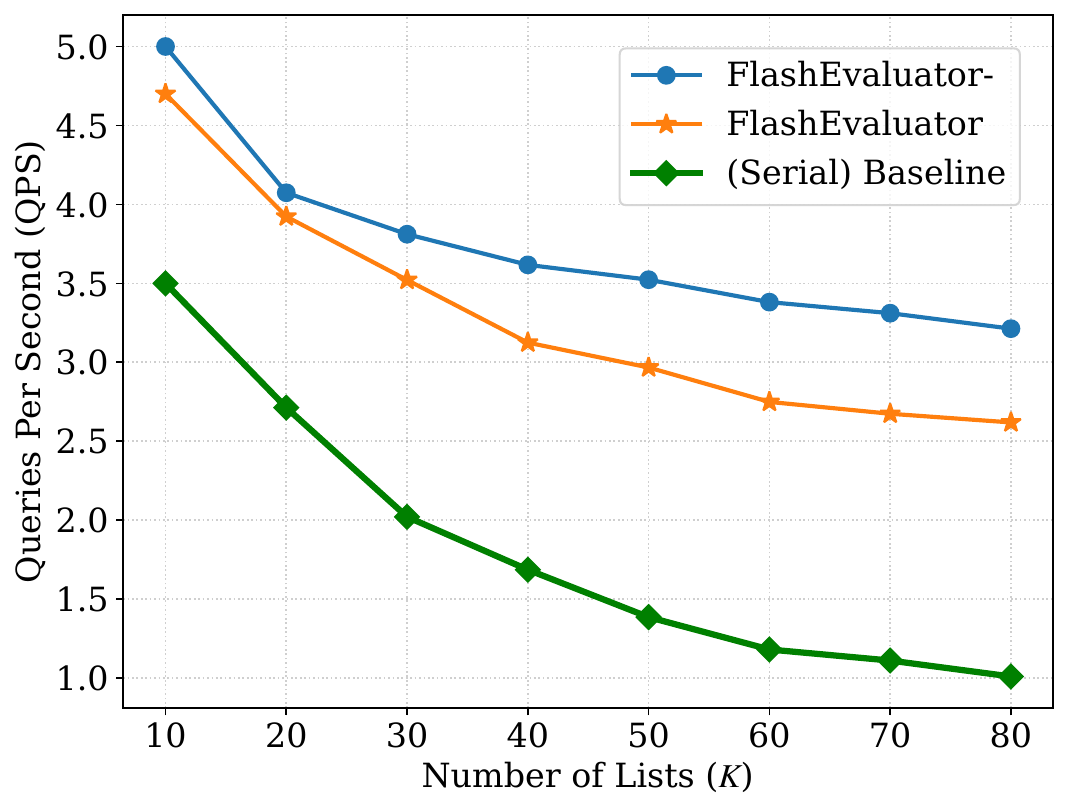}
        \caption{Average QPS of inference on Kuaishou.}
        \label{fig:online_QPS}
    \end{subfigure}
    \hfill
    \begin{subfigure}[t]{0.3\textwidth}
        \centering
        \includegraphics[width=\textwidth]{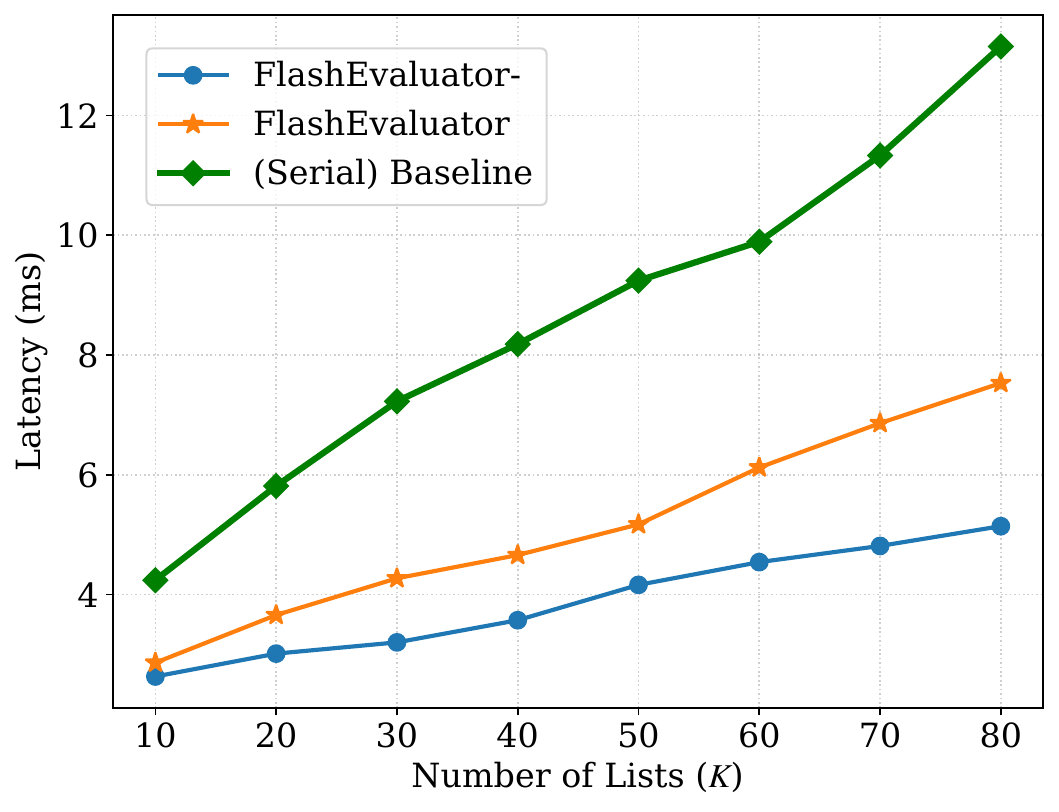}
        \caption{Average latency of inference on Kuaishou.}
        \label{fig:online_latency}
    \end{subfigure}
    \hfill
    \begin{subfigure}[t]{0.3\textwidth}
    \centering
    \includegraphics[width=\textwidth]{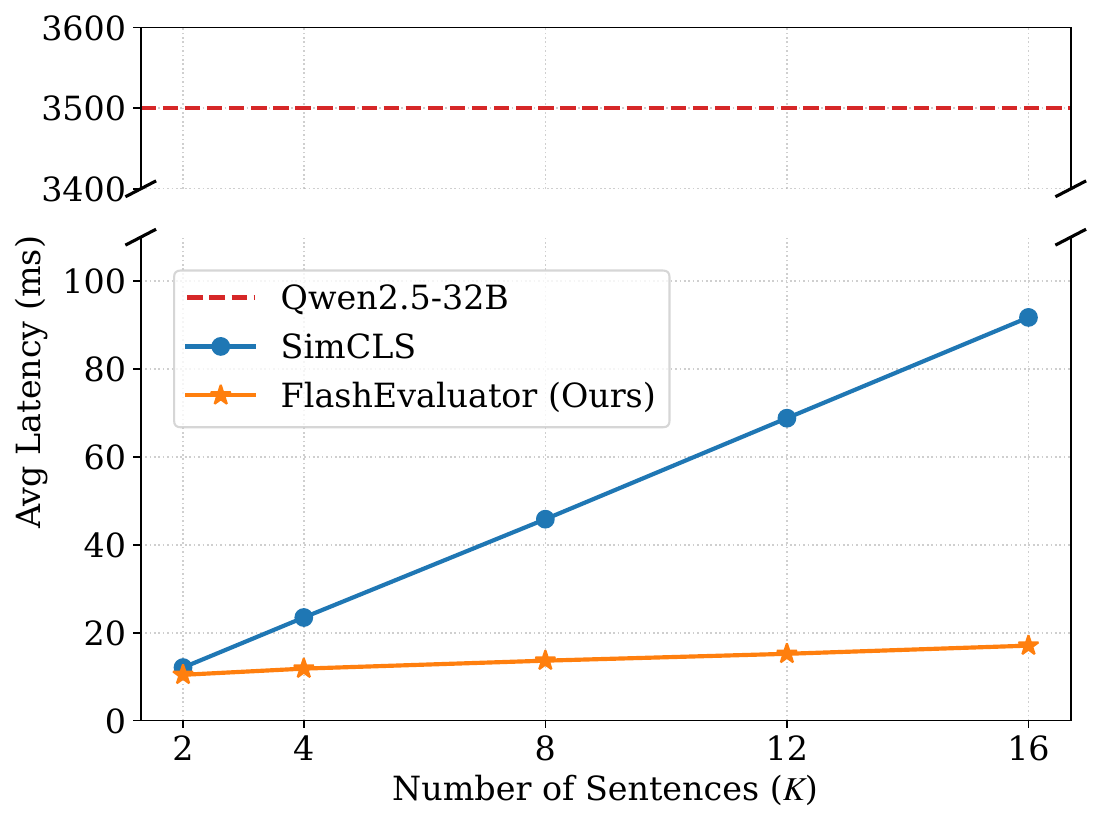}
    \caption{Average latency of inference process on CNN/DM dataset.}
    \label{fig:number_of_sentences}
    \end{subfigure}
    \caption{Comparison of efficiency}
\end{figure*}

\subsubsection{Baselines}

For recommendation, we compare FlashEvaluator with representative reranking methods, including generator-only methods \textbf{DNN}, \textbf{DCN}, \textbf{Seq2Slate}~\cite{bello2018seq2slate}, \textbf{DLCM}~\cite{ai2018learning}, \textbf{PRM}~\cite{pei2019personalized}, and \textbf{SetRank}~\cite{pang2020setrank}; the evaluator-only method \textbf{YOLOR}~\cite{wang2025yolor}; and generator-evaluator methods \textbf{NAR4Rec}~\cite{ren2024nar4rec} and \textbf{PIER}~\cite{shi2023pier}. To isolate the contribution of the evaluator, we retain the corresponding generators and replace their original evaluators with FlashEvaluator, denoted as \textbf{NAR4Rec+Flash} and \textbf{PIER+Flash}, respectively.

For text summarization, we compare with \textbf{SimCLS}~\cite{liu2021simcls} and \textbf{RankGPT}~\cite{sun2023chatgpt} with \textbf{Qwen2.5-32B}~\cite{yang2025qwen3} as backbone. Detailed descriptions and configurations are provided in Appendix~\ref{appdix:baselines}.

\subsubsection{Metrics and Implementation}

For recommendation, we report \textbf{NDCG@6}, \textbf{Precision@6}, \textbf{Recall@6}, and \textbf{F1@6}. For text summarization, we report \textbf{ROUGE-1, ROUGE-2, ROUGE-L}, and BERTScore~\textbf{(BS)}. Efficiency is evaluated using average inference latency and queries per second while varying the number of candidate sequences $K$. For the online A/B test, we optimize the long-term \text{retention} metrics (i.e., 7-day User Lifetime), the short-term \text{engagement metrics} (i.e., APP Duration Time and number of Active Users), , and diversity-oriented \text{ecosystem metric} (i.e., Cold-start Exposure).
We also quantify the computational advantage of our method by measuring the Average Inference \textbf{Latency} per query and the Throughput (Queries Per Second, \textbf{QPS}), compared to online baseline which utilizes one-by-one evaluation. Unless otherwise specified, the embedding dimension is 64, and all models are optimized using Adam~\cite{kingma2015adam}. We relegate implementation, baseline, online environment and data construction to Appendices~\ref{appdix:datasets}–\ref{appdix:implementation}.

\begin{table}[ht]
\centering
\caption{Online A/B testing results over a one-week period on Kuaishou APP. The proposed \textbf{FlashEvaluator} is compared against the production baseline. All reported improvements are statistically significant ($p < 0.05$).}
\label{tab:online_ab}
\begin{tabular}{llc}
\toprule
\textbf{Category} & \textbf{Metric} & \textbf{Relative Lift} \\
\midrule
\multirow{1}{*}{Retention} 
 & 7-day Lifetime (LT7) & $+0.039\%$ \\
\midrule
\multirow{2}{*}{Engagement} 
 & Duration Time(per User)  & $+0.142\%$ \\
 & Active User & $+0.077\%$\\
\midrule
\multirow{1}{*}{Ecosystem} 
 & Cold-start Exposure & $+2.507\%$ \\
\bottomrule
\end{tabular}
\end{table}


\subsection{Performance on Recommendation}
\label{sec:experiments_recommendation}

\paragraph{Offline results (RQ1).} Table~\ref{tab:main_results} reports the results on ML-1M, Amazon-Books, and RecFlow. FlashEvaluator consistently improves both generator backbones (NAR4Rec and PIER) across all three datasets. These results demonstrate that the proposed evaluator is compatible with different generators and benefits both public interaction benchmarks and industrial multi-stage recommendation. Component-level ablations and the effects of cross-sequence interaction and listwise supervision are reported in Appendix~\ref{appdix:ablation}. We further vary the number of generated candidate sequences in Appendix~\ref{appdix:candidate_expansion}, showing that the efficiency gain can be converted into improved recommendation quality through a larger candidate search space.

\paragraph{Online deployment (RQ3 \& RQ4).} Table~\ref{tab:online_ab} presents the seven-day A/B test on Kuaishou. FlashEvaluator yields statistically significant improvements in retention, engagement, and cold-start exposure, indicating that it integrates effectively with the production multi-objective prediction. At $K=50$, FlashEvaluator reduces inference latency by 44\% and increases QPS by 114\% relative to the production baseline. Considering that Kuaishou has over 400 million DAUs and each spends an average of more than 2 hours per day on the platform, this improvement is highly significant. Figures~\ref{fig:online_QPS} and~\ref{fig:online_latency} further show that its efficiency advantage becomes more pronounced as $K$ increases. Besides, removing cross-sequence interaction provides slightly higher efficiency.

\subsection{Performance on Text Summarization}
\label{sec:experiments_summarization}

Table~\ref{tab:summarization_rerank} reports the CNN/DailyMail results (RQ2). With the same $K=16$ candidates, FlashEvaluator achieves competitive performance against SimCLS and RankGPT while requiring substantially less reranking computation. We further enlarge its candidate set to $K=64$ under a latency budget comparable to SimCLS. The larger search space improves all metrics and allows FlashEvaluator to outperform SimCLS across the three generator backbones, demonstrating that its efficiency advantage can be converted into better candidate coverage and final selection quality. Figure~\ref{fig:number_of_sentences} confirms that FlashEvaluator scales substantially better with $K$ than candidate-factorized rerankers (RQ3).




\section{Conclusion}
\label{sec:conclusion}

We study the evaluator bottleneck in Generator-Evaluator systems and propose FlashEvaluator, a joint sequence-level evaluator that combines request-local QKV-Cache, index-based sequence assembly, and cross-sequence interaction. QKV-Cache reuses context-side key/value representations and request-conditioned candidate-side representations across candidate sequences, reducing repeated computation before sequence-specific modeling. Our analysis shows that, under a $K$-stable global capacity condition, joint score construction avoids the explicit $\sqrt{K}$ factor in a generic candidate-factorized generalization bound; in repeated-element and encoding-dominant regimes, QKV-Cache reduces the dominant candidate-side encoding term from $Kl$ occurrences to $U$ distinct elements. Experiments on three recommendation datasets and CNN/DailyMail demonstrate that FlashEvaluator consistently improves different Generator--Evaluator backbones while achieving favorable latency and throughput. In Kuaishou's production recommendation system, FlashEvaluator reduces inference latency by 44\% and increases QPS by 114\% at $K=50$, together with statistically significant gains in retention, engagement, and ecosystem metrics. On text summarization, FlashEvaluator achieves competitive quality with the same candidate set and converts its efficiency advantage into improved selection quality by evaluating more candidates with lower latency than SimCLS.



\bibliographystyle{ACM-Reference-Format}
\bibliography{sample-base}

\newpage
\appendix
\section{Theoretical Results}

\subsection{Detailed Proof of the Generalization Bounds}
\label{appdix:generalization}

We first recall the standard Rademacher generalization bound for bounded scalar-valued function classes\citep{bartlett2001rademacher}.

\begin{lemma}[Uniform deviation bound]
\label{lem:rad-general}
Let $\mathcal{G}$ be a class of functions mapping $(z,y)$ into $[0,1]$, and define
\[
R(g)=\mathbb{E}_{(z,y)\sim\mathcal{D}}[g(z,y)],\quad\widehat{R}_S(g)=\frac{1}{n}\sum_{t=1}^{n}g(z^{(t)},y^{(t)}).
\]
For any $\delta\in(0,1)$, with probability at least $1-\delta$ over an i.i.d.\ sample $S=\{(z^{(t)},y^{(t)})\}_{t=1}^{n}\sim\mathcal{D}^{n}$,
\begin{equation}
\sup_{g\in\mathcal{G}}\left(R(g)-\widehat{R}_S(g)\right)\leq2\mathfrak{R}_n(\mathcal{G})+3\sqrt{\frac{\log(2/\delta)}{2n}}.
\end{equation}
\end{lemma}

For a vector-valued function class $\mathcal{H}\subseteq\{h:\mathcal{Z}\rightarrow\mathbb{R}^{K}\}$,
we use the empirical vector Rademacher complexity
\begin{equation}
\widehat{\mathfrak{R}}_S(\mathcal{H})=\mathbb{E}_{\epsilon}\left[\sup_{h\in\mathcal{H}}\frac{1}{n}\sum_{t=1}^{n}\left\langle\boldsymbol{\epsilon}^{(t)},h(z^{(t)})\right\rangle\right],
\end{equation}
where $\boldsymbol{\epsilon}^{(t)}=(\epsilon_{t,1},\ldots,\epsilon_{t,K})$ contains independent Rademacher variables. Its expected counterpart is $\mathfrak{R}_n(\mathcal{H})=\mathbb{E}_S\left[\widehat{\mathfrak{R}}_S(\mathcal{H})\right]$.

For the two evaluator classes, define
\begin{align*}
\mathcal{G}_{\mathrm{joint}}&=\left\{(z,y)\mapsto\ell\bigl(E_\theta(z),y\bigr):\theta\in\Theta_{\mathrm{joint}}\right\},\\
\mathcal{G}_{\mathrm{ind}}&=\left\{(z,y)\mapsto\ell\bigl(E_\theta^{\mathrm{ind}}(z),y\bigr):\theta\in\Theta_{\mathrm{ind}}\right\}.
\end{align*}

We first provide a sufficient condition for Assumption~\ref{ass:capacity}.
\begin{lemma}
\label{appdix:lem:joint-capacity}
For every fixed finite collection
$S_m=\{z^{(1)},\ldots,z^{(m)}\}$ of $m$ valid inputs, suppose the
joint score class admits the decomposition
\[
\mathcal{H}_{\mathrm{joint}}
=
\Psi\circ\mathcal{H}_{\mathrm{hid}},
\quad
\mathcal{H}_{\mathrm{hid}}
=
\{z\mapsto h_\theta(z):\theta\in\Theta_{\mathrm{joint}}\},
\]
where $\Psi$ is $\Lambda$-Lipschitz under the corresponding
Hilbert-space norm, with $\Lambda$ independent of $K$, and
\[
\widehat{\mathfrak{R}}_{S_m}
(\mathcal{H}_{\mathrm{hid}})
\leq
\frac{C_{\mathrm{hid}}}{\sqrt{m}}.
\]
Then
\[
\widehat{\mathfrak{R}}_{S_m}
(\mathcal{H}_{\mathrm{joint}})
\leq
c_{\mathrm{vc}}
\frac{\Lambda C_{\mathrm{hid}}}{\sqrt{m}},
\]
where $c_{\mathrm{vc}}$ is the universal constant from the vector
contraction inequality.
\end{lemma}

\begin{proof}
For the fixed input collection $S_m$, applying the vector
contraction inequality to
$\mathcal{H}_{\mathrm{joint}}
=\Psi\circ\mathcal{H}_{\mathrm{hid}}$ gives
\[
\widehat{\mathfrak{R}}_{S_m}
(\mathcal{H}_{\mathrm{joint}})
\leq
c_{\mathrm{vc}}\Lambda
\widehat{\mathfrak{R}}_{S_m}
(\mathcal{H}_{\mathrm{hid}}).
\]
Substituting the assumed empirical complexity bound on
$\mathcal{H}_{\mathrm{hid}}$ yields
\[
\widehat{\mathfrak{R}}_{S_m}
(\mathcal{H}_{\mathrm{joint}})
\leq
c_{\mathrm{vc}}
\frac{\Lambda C_{\mathrm{hid}}}{\sqrt{m}},
\]
which proves the result.
\end{proof}

For notational simplicity, we absorb the universal constant
$c_{\mathrm{vc}}$ into $\Lambda$ in the subsequent results.
Taking expectation over $S_m$ also yields the corresponding
expected-complexity bound
\[
\mathfrak{R}_m(\mathcal{H}_{\mathrm{joint}})
\leq
\frac{\Lambda C_{\mathrm{hid}}}{\sqrt{m}}.
\]

\begin{lemma}[Vector contraction]
\label{lem:contraction}
Under Assumption~\ref{ass:loss}, for $a\in\{\mathrm{joint},\mathrm{ind}\}$,
\begin{equation}
\mathfrak{R}_n(\mathcal{G}_{a})\leq\sqrt{2}L_\ell\mathfrak{R}_n(\mathcal{H}_{a}).
\end{equation}
\end{lemma}
\begin{proof}
For each fixed label $y^{(t)}$, the map $\mathbf{s}\mapsto\ell(\mathbf{s},y^{(t)})$ is $L_\ell$-Lipschitz with respect to the Euclidean norm. Subtracting the constant $\ell(\mathbf{0},y^{(t)})$ does not change the Rademacher complexity. The result therefore follows directly from the vector contraction inequality of \citet{maurer2016vector}. Taking expectation over the sample completes the proof.
\end{proof}

\subsubsection{Joint Evaluator}
We first derive the bound for FlashEvaluator.
\begin{proposition}[Generalization bound for joint evaluator]
\label{prop:joint-bound}
Suppose Assumptions~\ref{ass:loss} and~\ref{ass:capacity} hold. For any $\delta\in(0,1)$, with probability at least $1-\delta$, the following holds simultaneously for every $\theta\in\Theta_{\mathrm{joint}}$:
\begin{equation}
\mathcal{R}(\theta)\leq\widehat{\mathcal{R}}_S(\theta)+\frac{2\sqrt{2}L_\ell\Lambda C_{\mathrm{hid}}}{\sqrt{n}}+3\sqrt{\frac{\log(2/\delta)}{2n}}.
\end{equation}
\end{proposition}

\begin{proof}
Applying Lemma~\ref{lem:contraction} and
Assumption~\ref{ass:capacity} gives
\[
\mathfrak{R}_n(\mathcal{G}_{\mathrm{joint}})\leq\sqrt{2}L_\ell\mathfrak{R}_n(\mathcal{H}_{\mathrm{joint}})\leq\sqrt{2}L_\ell\frac{\Lambda C_{\mathrm{hid}}}{\sqrt{n}}.
\]
Substituting this inequality into Lemma~\ref{lem:rad-general} yields
\[
\begin{aligned}
\mathcal{R}(\theta)-\widehat{\mathcal{R}}_S(\theta)&\leq2\mathfrak{R}_n(\mathcal{G}_{\mathrm{joint}})+3\sqrt{\frac{\log(2/\delta)}{2n}}\\
&\leq\frac{2\sqrt{2}L_\ell\Lambda C_{\mathrm{hid}}}{\sqrt{n}}+3\sqrt{\frac{\log(2/\delta)}{2n}},
\end{aligned}
\]
uniformly over $\theta\in\Theta_{\mathrm{joint}}$.
\end{proof}

\subsubsection{Independent Evaluator}

For each training instance $z^{(t)}=(x^{(t)},C^{(t)},L_{1:K}^{(t)})$, define the corresponding per-sequence inputs as
\[
z^{(t,k)}
=
(x^{(t)},L_k^{(t)}),
\quad
t\in[n],\quad k\in[K].
\]
Let $\widetilde{S}=\left\{z^{(t,k)}:t\in[n],\,k\in[K]\right\}$ denote the resulting flattened collection of $nK$ per-sequence inputs. The elements of $\widetilde{S}$ need not be independent because the $K$ candidate sequences from the same instance share the request context. 
\begin{lemma}[Score-class complexity]
\label{lem:ind-complexity}
For every fixed sample $S$,
\begin{equation}
\widehat{\mathfrak{R}}_S\bigl(\mathcal{H}_{\mathrm{ind}}\bigr)=K\,\widehat{\mathfrak{R}}_{\widetilde{S}}\bigl(\mathcal{H}_{\mathrm{base}}\bigr).
\end{equation}
Consequently, under Assumption~\ref{ass:capacity},
\begin{equation}
\mathfrak{R}_n\bigl(\mathcal{H}_{\mathrm{ind}}\bigr)\leq C_{\mathrm{base}}\sqrt{\frac{K}{n}}.
\end{equation}
\end{lemma}
\begin{proof}
By the definition of the independent evaluator,
\[
E_\theta^{\mathrm{ind}}(z^{(t)})
=
\bigl(
f_\theta(z^{(t,1)}),
\ldots,
f_\theta(z^{(t,K)})
\bigr).
\]
Therefore,
\[
\begin{aligned}
\widehat{\mathfrak{R}}_S
\bigl(\mathcal{H}_{\mathrm{ind}}\bigr)
&=
\mathbb{E}_{\epsilon}
\left[
\sup_{\theta\in\Theta_{\mathrm{ind}}}
\frac{1}{n}
\sum_{t=1}^{n}
\sum_{k=1}^{K}
\epsilon_{t,k}
f_\theta(z^{(t,k)})
\right]\\
&=
K\,
\mathbb{E}_{\epsilon}
\left[
\sup_{\theta\in\Theta_{\mathrm{ind}}}
\frac{1}{nK}
\sum_{t=1}^{n}
\sum_{k=1}^{K}
\epsilon_{t,k}
f_\theta(z^{(t,k)})
\right]\\
&=
K\,
\widehat{\mathfrak{R}}_{\widetilde{S}}
\bigl(\mathcal{H}_{\mathrm{base}}\bigr).
\end{aligned}
\]

The capacity condition is understood as a uniform empirical bound:
for every collection of $m$ inputs,
\[
\widehat{\mathfrak{R}}
\bigl(\mathcal{H}_{\mathrm{base}}\bigr)
\leq
\frac{C_{\mathrm{base}}}{\sqrt{m}}.
\]
Applying it to the fixed collection $\widetilde{S}$ of size $nK$
gives
\[
\widehat{\mathfrak{R}}_S
\bigl(\mathcal{H}_{\mathrm{ind}}\bigr)
\leq
K\frac{C_{\mathrm{base}}}{\sqrt{nK}}
=
C_{\mathrm{base}}
\sqrt{\frac{K}{n}}.
\]
Taking expectation over $S$ proves the result.
\end{proof}
\begin{corollary}[Loss-class complexity]
\label{cor:ind-scaling}
Under Assumptions~\ref{ass:loss} and~\ref{ass:capacity},
\begin{equation}
\mathfrak{R}_n
\bigl(\mathcal{G}_{\mathrm{ind}}\bigr)
\leq
\sqrt{2}L_\ell C_{\mathrm{base}}
\sqrt{\frac{K}{n}}.
\end{equation}
\end{corollary}
\begin{proof}
Combining Lemmas~\ref{lem:contraction} and
\ref{lem:ind-complexity} gives
\[
\mathfrak{R}_n
\bigl(\mathcal{G}_{\mathrm{ind}}\bigr)\leq
\sqrt{2}L_\ell
\mathfrak{R}_n
\bigl(\mathcal{H}_{\mathrm{ind}}\bigr)\leq
\sqrt{2}L_\ell C_{\mathrm{base}}
\sqrt{\frac{K}{n}}.
\]
\end{proof}
\begin{proposition}[Generalization bound for independent evaluator]
\label{prop:ind-bound}
Suppose Assumptions~\ref{ass:loss} and~\ref{ass:capacity} hold.
For any $\delta\in(0,1)$, with probability at least $1-\delta$,
the following holds simultaneously for every
$\theta\in\Theta_{\mathrm{ind}}$:
\begin{equation}
\mathcal{R}^{\mathrm{ind}}(\theta)\leq\widehat{\mathcal{R}}^{\mathrm{ind}}_S(\theta)+2\sqrt{2}L_\ell C_{\mathrm{base}}\sqrt{\frac{K}{n}}+3\sqrt{\frac{\log(2/\delta)}{2n}}.
\end{equation}
\end{proposition}

\begin{proof}
By Corollary~\ref{cor:ind-scaling} and
Lemma~\ref{lem:rad-general},
\[
\begin{aligned}
\mathcal{R}^{\mathrm{ind}}(\theta)
-
\widehat{\mathcal{R}}^{\mathrm{ind}}_S(\theta)
&\leq
2\mathfrak{R}_n(\mathcal{G}_{\mathrm{ind}})
+
3\sqrt{\frac{\log(2/\delta)}{2n}}\\
&\leq
2\sqrt{2}L_\ell C_{\mathrm{base}}
\sqrt{\frac{K}{n}}
+
3\sqrt{\frac{\log(2/\delta)}{2n}},
\end{aligned}
\]
uniformly over $\theta\in\Theta_{\mathrm{ind}}$.
\end{proof}

\begin{proof}[Proof of Theorem~\ref{thm:comparison}]
Apply Propositions~\ref{prop:joint-bound} and~\ref{prop:ind-bound}
with confidence parameter $\delta/2$. By the union bound, both
inequalities hold simultaneously with probability at least
$1-\delta$:
\[
\begin{aligned}
\mathcal{R}(\theta)
-
\widehat{\mathcal{R}}_S(\theta)
&\leq
\frac{2\sqrt{2}L_\ell\Lambda C_{\mathrm{hid}}}{\sqrt{n}}
+
3\sqrt{\frac{\log(4/\delta)}{2n}},\\
\mathcal{R}^{\mathrm{ind}}(\theta)
-
\widehat{\mathcal{R}}^{\mathrm{ind}}_S(\theta)
&\leq
2\sqrt{2}L_\ell C_{\mathrm{base}}
\sqrt{\frac{K}{n}}
+
3\sqrt{\frac{\log(4/\delta)}{2n}}.
\end{aligned}
\]

The ratio between the leading capacity terms is
\[
\frac{
2\sqrt{2}L_\ell C_{\mathrm{base}}
\sqrt{K/n}
}{
2\sqrt{2}L_\ell\Lambda C_{\mathrm{hid}}/\sqrt{n}
}
=
\frac{C_{\mathrm{base}}}
{\Lambda C_{\mathrm{hid}}}
\sqrt{K}.
\]
Hence, when
\[
\frac{C_{\mathrm{base}}}
{\Lambda C_{\mathrm{hid}}}
=
\Theta(1),
\]
the ratio scales as $\Theta(\sqrt{K})$.
\end{proof}
\subsection{Detailed Analysis of Robustness}
\label{appdix:robustness}

We next analyze \textbf{robustness} to generator-induced \textbf{SSB} between the training and test environments. Unlike the generalization analysis, the following analysis concerns the training regimes: pointwise squared loss and listwise softmax cross-entropy. 

Let $\mathcal{D}_{\mathrm{train}}$ and $\mathcal{D}_{\mathrm{test}}$ denote the distributions induced by the training and test generators, respectively. We assume that the
conditional feedback mechanisms below are defined on the support of
$\mathcal{D}_{\mathrm{test}}$.

\begin{assumption}[Decomposable additive feedback shift]
\label{ass:ssb-decomp}
For each environment $\tau\in\{\mathrm{train},\mathrm{test}\}$, the scalar feedback associated with candidate sequence $L_k$ satisfies
\[
r_k=u_k^\star(z)+\bar{\mu}_{\tau}(z)+\bar{\nu}_{\tau,k}(z)+\varepsilon_{\tau,k},\quad\mathbb{E}[\varepsilon_{\tau,k}\mid z]=0,
\]
where $\sum_{k=1}^{K}\bar{\nu}_{\tau,k}(z)=0$. Here, $\bar{\mu}_{\tau}(z)$ is a shared candidate-set-level shift,
whereas $\bar{\nu}_{\tau,k}(z)$ is a candidate-specific distortion.
\end{assumption}

Define $\Delta\bar{\mu}=\bar{\mu}_{\mathrm{train}}-\bar{\mu}_{\mathrm{test}}, \Delta\bar{\nu}_k=\bar{\nu}_{\mathrm{train},k}-\bar{\nu}_{\mathrm{test},k}$ and let $\Delta\bar{\boldsymbol{\nu}}=(\Delta\bar{\nu}_1,\ldots,\Delta\bar{\nu}_K)^\top$. To isolate sensitivity to the learning objective, we conduct a
Bayes-level analysis and assume that the corresponding score classes
are sufficiently rich to realize the population-optimal predictors.
This removes approximation error and therefore gives the predictor an oracle comparison.

\subsubsection{Pointwise Regression}

Consider the pointwise squared risk
\begin{equation}
\mathcal{R}_{\mathrm{pt}}^{\tau}(\bm{f})=
\mathbb{E}_{z\sim\mathcal{D}_{\tau}}
\left[
\frac{1}{K}
\sum_{k=1}^{K}
\mathbb{E}
\left[
\bigl(f_k(z)-r_k\bigr)^2
\mid z
\right]
\right].
\end{equation}
Under Assumption~\ref{ass:ssb-decomp}, its Bayes-optimal predictor is
\begin{equation}
f_{\tau,k}^{\star}(z)=u_k^\star(z)+\bar{\mu}_{\tau}(z)+\bar{\nu}_{\tau,k}(z).
\end{equation}
We define the pointwise excess test risk induced by the environment
shift as
\begin{equation}
\Delta_{\mathrm{SSB}}^{\mathrm{ind}}=\mathcal{R}_{\mathrm{pt}}^{\mathrm{test}}\bigl(\bm{f}_{\mathrm{train}}^\star\bigr)-\mathcal{R}_{\mathrm{pt}}^{\mathrm{test}}\bigl(\bm{f}_{\mathrm{test}}^\star\bigr).
\end{equation}

\begin{proposition}[Pointwise excess risk]
\label{prop:ind-ssb}
Under Assumption~\ref{ass:ssb-decomp},
\begin{equation}
\Delta_{\mathrm{SSB}}^{\mathrm{ind}}
=
\mathbb{E}_{z\sim\mathcal{D}_{\mathrm{test}}}
\left[
(\Delta\bar{\mu})^2
+
\frac{1}{K}
\sum_{k=1}^{K}
(\Delta\bar{\nu}_k)^2
\right].
\end{equation}
\end{proposition}

\begin{proof}
For squared loss, the standard Bayes-risk decomposition
\[
\mathcal{R}_{\mathrm{pt}}^{\mathrm{test}}
\bigl(\bm{f}_{\mathrm{train}}^\star\bigr)
-
\mathcal{R}_{\mathrm{pt}}^{\mathrm{test}}
\bigl(\bm{f}_{\mathrm{test}}^\star\bigr) =
\mathbb{E}_{z\sim\mathcal{D}_{\mathrm{test}}}
\left[
\frac{1}{K}
\sum_{k=1}^{K}
\left(
f_{\mathrm{train},k}^\star(z)
-
f_{\mathrm{test},k}^\star(z)
\right)^2
\right].
\]
Substituting the two Bayes predictors yields
\[
\Delta_{\mathrm{SSB}}^{\mathrm{ind}}
=
\mathbb{E}_{z\sim\mathcal{D}_{\mathrm{test}}}
\left[
\frac{1}{K}
\sum_{k=1}^{K}
\bigl(
\Delta\bar{\mu}
+
\Delta\bar{\nu}_k
\bigr)^2
\right].
\]
Expanding the square and using
$\sum_{k=1}^{K}\Delta\bar{\nu}_k=0$ gives
\[
\frac{1}{K}
\sum_{k=1}^{K}
\bigl(
\Delta\bar{\mu}
+
\Delta\bar{\nu}_k
\bigr)^2
=
(\Delta\bar{\mu})^2
+
\frac{1}{K}
\sum_{k=1}^{K}
(\Delta\bar{\nu}_k)^2,
\]
which proves the claim.
\end{proof}

Proposition~\ref{prop:ind-ssb} shows that pointwise regression
absorbs both the shared shift and the candidate-specific distortion.
In particular, even a shift that adds the same value to every
candidate target contributes the nonzero term
$(\Delta\bar{\mu})^2$.

\subsubsection{Listwise Softmax}

We next consider listwise supervision. To obtain a smooth local characterization, we adopt the following stylized multinomial-choice model. Under environment $\tau$, let
the conditional label distribution be generated from the shifted
utilities:
\[
p_{\tau,k}(z)
=
\frac{
\exp\bigl(
u_k^\star(z)
+
\bar{\mu}_{\tau}(z)
+
\bar{\nu}_{\tau,k}(z)
\bigr)
}{
\sum_{j=1}^{K}
\exp\bigl(
u_j^\star(z)
+
\bar{\mu}_{\tau}(z)
+
\bar{\nu}_{\tau,j}(z)
\bigr)
}.
\]
Because the shared shift is identical for all candidates,
\[
p_{\tau,k}(z)
=
\frac{
\exp\bigl(
u_k^\star(z)+\bar{\nu}_{\tau,k}(z)
\bigr)
}{
\sum_{j=1}^{K}
\exp\bigl(
u_j^\star(z)+\bar{\nu}_{\tau,j}(z)
\bigr)
}.
\]
Thus, $\bar{\mu}_{\tau}$ cancels exactly before model estimation. For a score vector $\bm{s}(z)$, define the listwise risk
\begin{equation}
\mathcal{R}_{\mathrm{list}}^{\tau}(\bm{s})
=
\mathbb{E}_{z\sim\mathcal{D}_{\tau}}
\left[
-\sum_{k=1}^{K}
p_{\tau,k}(z)
\log
\frac{\exp(s_k(z))}
{\sum_{j=1}^{K}\exp(s_j(z))}
\right].
\end{equation}

\begin{lemma}[Bayes scores and shift invariance]
\label{lem:joint-bayes}
A score vector minimizes
$\mathcal{R}_{\mathrm{list}}^{\tau}$ if and only if
\[
s_{\tau,k}^\star(z)
=
u_k^\star(z)
+
\bar{\nu}_{\tau,k}(z)
+
c(z),
\quad
k\in[K],
\]
where $c(z)$ is an arbitrary scalar function. In particular, the
Bayes-optimal softmax distribution is independent of
$\bar{\mu}_{\tau}$.
\end{lemma}
\begin{proof}
Cross-entropy is minimized when the predicted conditional
distribution equals $\bm{p}_{\tau}(z)$. Softmax scores are
identifiable only up to a common additive constant, which gives the
stated family of minimizers.
\end{proof}
Therefore, we define the listwise excess risk and give
\begin{equation}
\Delta_{\mathrm{SSB}}^{\mathrm{joint}}
=
\mathcal{R}_{\mathrm{list}}^{\mathrm{test}}
\bigl(\bm{s}_{\mathrm{train}}^\star\bigr)
-
\mathcal{R}_{\mathrm{list}}^{\mathrm{test}}
\bigl(\bm{s}_{\mathrm{test}}^\star\bigr).
\end{equation}
\begin{proposition}[Listwise excess risk]
\label{prop:joint-ssb}
Under Assumption~\ref{ass:ssb-decomp},
\begin{equation}
\Delta_{\mathrm{SSB}}^{\mathrm{joint}}
=
\mathbb{E}_{z\sim\mathcal{D}_{\mathrm{test}}}
\left[
\mathrm{KL}
\left(
\bm{p}_{\mathrm{test}}(z)
\,\middle\|\,
\bm{p}_{\mathrm{train}}(z)
\right)
\right].
\end{equation}
Moreover, let
\[
\bm{s}_{\mathrm{test}}
=
\bm{u}^{\star}
+
\bar{\boldsymbol{\nu}}_{\mathrm{test}},
\quad
\bm{s}_{\mathrm{train}}
=
\bm{s}_{\mathrm{test}}
+
\Delta\bar{\boldsymbol{\nu}}.
\]
Then the KL divergence admits the exact representation
\begin{equation}
\mathrm{KL}
\left(
\bm{p}_{\mathrm{test}}
\,\middle\|\,
\bm{p}_{\mathrm{train}}
\right)
=
\int_{0}^{1}
(1-t)\,
\Delta\bar{\boldsymbol{\nu}}^{\top}
\bm{H}_t
\Delta\bar{\boldsymbol{\nu}}
\,dt,
\end{equation}
where $\bm{H}_t=\operatorname{diag}(\bm{p}_t)-\bm{p}_t\bm{p}_t^\top$ and  $\bm{p}_t=\operatorname{softmax}\left(\bm{s}_{\mathrm{test}}+t\Delta\bar{\boldsymbol{\nu}}\right)$. Consequently, for a sufficiently small distortion gap,
\[
\Delta_{\mathrm{SSB}}^{\mathrm{joint}}
=
\frac{1}{2}
\mathbb{E}_{z\sim\mathcal{D}_{\mathrm{test}}}
\left[
\operatorname{Var}_{k\sim\bm{p}_{\mathrm{test}}}
\bigl(
\Delta\bar{\nu}_k
\bigr)
\right]
+
o\left(
\mathbb{E}_{z\sim\mathcal{D}_{\mathrm{test}}}
\left[
\|\Delta\bar{\boldsymbol{\nu}}\|_2^2
\right]
\right).
\]
\end{proposition}

\begin{proof}
The difference between the cross-entropy attained by $\bm{p}_{\mathrm{train}}$ and that attained by $\bm{p}_{\mathrm{test}}$ is $\mathrm{KL}\left(\bm{p}_{\mathrm{test}}
\,\middle\|\,\bm{p}_{\mathrm{train}}\right)$, which proves the first statement. Let $A(\bm{s})=\log\sum_{k=1}^{K}\exp(s_k)$. Since
\[
\nabla A(\bm{s})
=
\operatorname{softmax}(\bm{s}),
\quad
\nabla^2 A(\bm{s})
=
\operatorname{diag}(\bm{p})
-
\bm{p}\bm{p}^{\top},
\]
the KL divergence can be written as
\[
\mathrm{KL}
\left(
\bm{p}_{\mathrm{test}}
\,\middle\|\,
\bm{p}_{\mathrm{train}}
\right)
=
A(\bm{s}_{\mathrm{train}})
-
A(\bm{s}_{\mathrm{test}}) -
\nabla A(\bm{s}_{\mathrm{test}})^\top
\left(
\bm{s}_{\mathrm{train}}
-
\bm{s}_{\mathrm{test}}
\right).
\]
Applying Taylor's theorem with integral remainder and using
$\bm{s}_{\mathrm{train}}-\bm{s}_{\mathrm{test}}
=\Delta\bar{\boldsymbol{\nu}}$ gives the exact integral form.

For any vector $\bm{v}$,
\[
\bm{v}^{\top}
\left(
\operatorname{diag}(\bm{p})
-
\bm{p}\bm{p}^{\top}
\right)
\bm{v}
=
\operatorname{Var}_{k\sim\bm{p}}(v_k).
\]
Taking the second-order expansion around
$\bm{s}_{\mathrm{test}}$ therefore yields
\[
\mathrm{KL}
\left(
\bm{p}_{\mathrm{test}}
\,\middle\|\,
\bm{p}_{\mathrm{train}}
\right)
=
\frac{1}{2}
\operatorname{Var}_{k\sim\bm{p}_{\mathrm{test}}}
\bigl(
\Delta\bar{\nu}_k
\bigr)
+
o\left(
\|\Delta\bar{\boldsymbol{\nu}}\|_2^2
\right).
\]
Taking expectation over the test distribution completes the proof.
\end{proof}

The expression confirms that the listwise excess risk contains no
term involving $\Delta\bar{\mu}$. Its local sensitivity is determined
only by the relative distortion across candidate sequences.

\subsection{Detailed Proof of Computational Complexity}
\label{appdix:complexity}

Let $C_{\mathrm{item}}$ denote the cost of constructing the
representation of one candidate element,
$C_{\mathrm{list}}$ the cost of aggregating one assembled candidate
sequence,
$C_{\mathrm{idx}}$ the cost of retrieving one element representation
through indexing,
$C_{\mathrm{att}}(K)$ the cost of cross-sequence attention over $K$
sequence representations, and
$C_{\mathrm{score}}$ the cost of scoring one sequence representation.

For clarity, the following analysis focuses on candidate-side
representation reuse. We omit the request-context encoding term,
because the candidate-factorized evaluator may repeat this computation
for different candidate sequences, whereas FlashEvaluator computes it
once per candidate instance. Including this term would only further
favor FlashEvaluator.

\subsubsection{Complexity of the Independent Evaluator}

The independent evaluator applies the same scoring model to
each candidate sequence separately. For one candidate sequence
$L_k$ of length $l$, the computational cost is
\[
T_{\mathrm{ind}}^{(k)}
=
\Theta\bigl(
lC_{\mathrm{item}}
+
C_{\mathrm{list}}
+
C_{\mathrm{score}}
\bigr).
\]
Summing over all $K$ candidate sequences gives
\[
\begin{aligned}
T_{\mathrm{ind}}
=
\sum_{k=1}^{K}
T_{\mathrm{ind}}^{(k)}
=
\Theta\bigl(
KlC_{\mathrm{item}}
+
KC_{\mathrm{list}}
+
KC_{\mathrm{score}}
\bigr).
\end{aligned}
\]

\subsubsection{Complexity of the Joint Evaluator}

FlashEvaluator separates candidate-element encoding from
sequence-specific processing. It first constructs one representation
for each of the $U$ distinct candidate elements:
\[
T_{\mathrm{item}}
=
\Theta\bigl(
UC_{\mathrm{item}}
\bigr).
\]
The $K$ candidate sequences are then assembled by retrieving the
shared representations of all $Kl$ element occurrences:
\[
T_{\mathrm{idx}}
=
\Theta\bigl(
KlC_{\mathrm{idx}}
\bigr).
\]
After indexing, each sequence is processed by the shared
intra-sequence aggregation module:
\[
T_{\mathrm{list}}
=
\Theta\bigl(
KC_{\mathrm{list}}
\bigr).
\]
Finally, FlashEvaluator models cross-sequence interactions over the
$K$ compressed sequence-level representations and predicts the
$K$ scores:
\[
T_{\mathrm{interact}}
=
\Theta\bigl(
C_{\mathrm{att}}(K)
+
KC_{\mathrm{score}}
\bigr).
\]
Although $C_{\mathrm{att}}(K)$ may grow quadratically with $K$ under
dense self-attention, it is applied only to the low-dimensional
sequence representations obtained after candidate-element encoding
and intra-sequence aggregation. Consequently, over the evaluated
range of $K$, its practical cost is small relative to the dominant
candidate-element representation construction.

Combining these components yields
\[
\label{eq:Tnew}
T_{\mathrm{joint}}
=
\Theta\Bigl(
UC_{\mathrm{item}}
+
KlC_{\mathrm{idx}}
+
KC_{\mathrm{list}}
+
C_{\mathrm{att}}(K)
+
KC_{\mathrm{score}}
\Bigr).
\]

\subsubsection{Comparison and Scalability}

We now prove Proposition~\ref{prop:complexity-advantage} under the
encoding-dominant regime, where candidate-element representation
construction is substantially more expensive than indexing,
sequence-level aggregation, score prediction, and cross-sequence
interaction.

\begin{proof}[Proof of Proposition~\ref{prop:complexity-advantage}]
Using $KlC_{\mathrm{item}}$ as the dominant term of
$T_{\mathrm{ind}}$, we obtain
\[
\begin{aligned}
\frac{T_{\mathrm{joint}}}{T_{\mathrm{ind}}}
&\approx
\frac{
UC_{\mathrm{item}}
+
KlC_{\mathrm{idx}}
+
KC_{\mathrm{list}}
+
C_{\mathrm{att}}(K)
+
KC_{\mathrm{score}}
}{
KlC_{\mathrm{item}}
+
KC_{\mathrm{list}}
+
KC_{\mathrm{score}}
}\\
&\approx
\frac{1}{\rho}
+
\mathcal{O}(\epsilon).
\end{aligned}
\]
Here, $\epsilon$ summarizes the costs of indexing, intra-sequence
aggregation, cross-sequence interaction, and score prediction,
normalized by the dominant candidate-element encoding cost. In
particular, $C_{\mathrm{att}}(K)$ is included in $\epsilon$. Although
dense cross-sequence attention has quadratic complexity in $K$, it
operates on compressed low-dimensional sequence representations and
accounts for only a small fraction of the total computation in our
serving regime. Therefore, $\epsilon$ remains small and the ratio is
dominated by $1/\rho$. The computational gain consequently increases
with the amount of element reuse across the candidate sequences.
\end{proof}

The result does not imply that the total cost of FlashEvaluator is
independent of $K$. Indexing, intra-sequence aggregation, and score
prediction remain at least linear in $K$, while
$C_{\mathrm{att}}(K)$ may grow quadratically when dense
self-attention is used. Nevertheless, cross-sequence attention is
performed only over $K$ compressed sequence-level representations
and remains substantially less expensive than repeated
candidate-element representation construction over the evaluated
range of $K$. Consequently, the approximation by $1/\rho$ applies
in the encoding-dominant serving regime targeted by FlashEvaluator.
\section{Details of Experiments}

The code will be released upon acceptance.

\subsection{Datasets}
\label{appdix:datasets}

\textbf{ML-1M}~\citep{harper2016movielens} and \textbf{Amazon-Books}~\citep{he2016ups} provide timestamped user--item interactions but do not contain request-level candidate pools from an upstream ranking stage. We therefore use \textbf{BPR-MF}~\citep{rendle2009bpr} to emulate first-stage retrieval. For each user, interactions are ordered chronologically, and the last six interactions are held out for final evaluation and excluded from both retrieval and reranking training. The remaining interactions are split at an $8{:}2$ ratio solely for training and validating BPR-MF, with the validation split used for model selection and early stopping. After BPR-MF is fixed, it scores all items in the full item corpus for each user.

The $8{:}2$ split is used only for fitting the retrieval model. All interactions preceding the six held-out test items are subsequently segmented into ordered subsequences of length six to construct reranking training instances. Because the public datasets do not record actual request-level exposure slates, each subsequence is treated as a pseudo-exposure sequence. Its six observed items are retained, while additional candidates are sampled from the user's top-200 BPR-MF ranking after removing duplicates. The resulting candidate pool is passed to the generator to construct $K$ candidate sequences of length six.

For final evaluation, the six held-out interactions form the target sequence. We combine these target items with the user's BPR-MF ranking, preserve all target items, remove duplicates, and fill the remaining positions in descending BPR-MF score order until exactly 50 distinct candidate items are obtained. The fitted BPR-MF model, candidate pools, and generated candidate sequences are fixed and shared across all compared methods.

\textbf{RecFlow}~\citep{liu2024recflow} is a large-scale industrial dataset that records user behaviors and candidate information throughout a multi-stage recommendation pipeline. Following its official protocol, we use the second data period from February 5 to February 18, containing 3,308,233 requests, and intercept the pipeline at the reranking stage. For each request, the candidate pool consists of the top 120 items passed from the preceding ranking module, with their original ranking positions retained as input features.

The input includes video ID, category, and author-related attributes, together with the user's 50 most recent interactions. The reranking objective is to select an ordered sequence of six items from the 120 upstream candidates. Effective-view feedback is used as item-level supervision, and the same request splits, candidate pools, and input features are shared by all compared methods.

\textbf{CNN/DailyMail}~\citep{hermann2015teaching} is a standard benchmark for abstractive text summarization. We use \textbf{T5-base}~\citep{raffel2020exploring}, \textbf{BART-Large}~\citep{lewis2020bart}, and \textbf{Llama-3.1-8B-Instruct}~\citep{grattafiori2024llama} as generator backbones. For the controlled comparison, each source article is associated with $K=16$ candidate summaries, generated using beam search for T5 and BART and diverse sampling for Llama-3.1-8B-Instruct. All rerankers receive the same candidate summaries, and reference summaries are used to construct training targets and calculate offline metrics. FlashEvaluator at $K=64$ still incurs lower reranking latency (~40ms) than SimCLS at $K=16$ (~90ms), we additionally report an expanded-candidate setting in which FlashEvaluator selects from 64 summaries, examining whether its efficiency advantage can be converted into improved candidate coverage and selection quality.

For each recommendation candidate sequence
$L_k=(i_{k,1},\ldots,i_{k,l})$, let
$a_{k,j}\in\{0,1\}$ denote the observed feedback label of item
$i_{k,j}$. We define the proxy utility of the complete sequence as
\[
r_k=\sum_{j=1}^{l} a_{k,j},
\]
i.e., the number of positively labeled items contained in $L_k$.
The optimal-candidate label is then constructed as
\[
y=\min\arg\max_{k\in[K]} r_k,
\]
where the smallest candidate index is selected to resolve ties.
The pointwise variants are trained with item-level binary
cross-entropy using $a_{k,j}$ as the supervision target, whereas
the listwise variant uses the resulting sequence-level label $y$
in the softmax cross-entropy objective.

\paragraph{Online deployment.} FlashEvaluator is deployed in the reranking stage of the single-column recommendation feed on \textbf{Kuaishou}'s main app. The platform serves over 400 million daily active users, with an average daily time spent exceeding two hours per user, providing a large-scale and latency-sensitive production environment. The production baseline and FlashEvaluator are assigned to two disjoint traffic buckets, each covering 10\% of the total traffic, and the A/B test lasts for seven consecutive days. Both systems use the same candidate generators, input features, and multi-objective serving targets.

\subsection{Baselines}
\label{appdix:baselines}

\textbf{Recommendation.} We compare FlashEvaluator with the following recommendation baselines:
\begin{itemize}
    \item \textbf{DNN} and \textbf{DCN} independently estimate item-level relevance without explicitly modeling contextual dependencies within the final sequence.
    \item \textbf{Seq2Slate}~\citep{bello2018seq2slate} is an autoregressive sequence-to-sequence reranking method that generates an ordered list using a pointer-network architecture.
    \item \textbf{DLCM}~\citep{ai2018learning} applies a recurrent listwise context model to refine an initial ranking based on local contextual dependencies.
    \item \textbf{PRM}~\citep{pei2019personalized} uses Transformer-based self-attention to model interactions among items in a personalized reranking list.
    \item \textbf{SetRank}~\citep{pang2020setrank} jointly models candidate items using a permutation-invariant set encoder.
    \item \textbf{YOLOR}~\citep{wang2025yolor} is an evaluator-only reranking method that uses a tree-structured context cache to reuse computation across candidate permutations.
    \item \textbf{NAR4Rec}~\citep{ren2024nar4rec} generates a complete reranking sequence non-autoregressively and employs an evaluator to select among generated candidates.
    \item \textbf{PIER}~\citep{shi2023pier} adopts a Generator--Evaluator pipeline that aggregates candidates from multiple generators and uses OCPM for permutation-level evaluation. We use the more recent \textbf{MultG}~\citep{yang2025comprehensive} as its candidate generator.
\end{itemize}

\textbf{Text-summarization.} We compare FlashEvaluator with \textbf{SimCLS}~\citep{liu2021simcls}, which learns a reference-free summary evaluator using contrastive supervision, and \textbf{RankGPT}~\citep{sun2023chatgpt}, which prompts a large language model to rank multiple candidate summaries. We use \textbf{Qwen2.5-32B}~\citep{yang2025qwen3} as the RankGPT backbone. 

\subsection{Implementation Details}
\label{appdix:implementation}

For recommendation, the target sequence length is fixed to six, the embedding dimension is set to 64, and Adam is used with a learning rate of $10^{-3}$. Unless otherwise specified, the batch size is 2048, and the number of generated candidate
sequences is set to $K=10$. Within each comparison, all methods use identical candidate pools, candidate generators, and generated candidate sequences.

For \textbf{ML-1M} and \textbf{Amazon-Books}, the $8{:}2$ split is used exclusively for \textbf{BPR-MF} training and validation. Once the retrieval model is selected, all interactions preceding the final six test interactions are used to construct reranking training instances. \textbf{BPR-MF} scores the full item corpus, and the same retrieval results are used by all reranking methods.

For efficiency evaluation, all latency and throughput measurements are conducted using the same hardware, numerical precision, and candidate sequences. The online comparison is against the deployed serial production baseline. Therefore, the reported results represent end-to-end deployment gains rather than a comparison with a batched independent evaluator.

For the online deployment, both FlashEvaluator and the production PIER-style baseline use Linear Attention~\cite{katharopoulos2020linearattention} to satisfy serving-latency constraints. The query context contains the user's 1,000 most recent interactions in chronological order. Candidate sequences are generated from a reranking pool of 60 videos, and $K=50$ sequences are passed to the evaluator. Online inference is served on cloud containers equipped with 120 CPU threads and 500 GB of memory.

For \textbf{CNN/DailyMail}, we use publicly available generator checkpoints from Hugging Face\footnote{\url{https://huggingface.co}}. Generator outputs are precomputed and fixed across reranking methods. In the latency-matched experiment, the candidate size of FlashEvaluator is increased until its reranking latency is comparable to that of SimCLS, allowing us to evaluate whether computational savings can be converted into improved candidate coverage and selection quality. The code will be released upon acceptance.

\subsection{Ablation Study}
\label{appdix:ablation}

The evaluator architecture and the learning objective are conceptually orthogonal: both candidate-factorized and joint evaluators can, in principle, be trained using either pointwise or listwise supervision. We therefore conduct an ablation study on RecFlow with the MultG generator to separately examine the contribution of cross-sequence interaction and the listwise objective.

\begin{itemize}
    \item \textbf{PIER + Flash $-$ CrossSeq} removes the sequence-level cross-interaction layer while retaining shared request/candidate encoding, index-based sequence assembly, and intra-sequence aggregation. Each sequence is scored without explicit interaction among sequence-level representations, while the shared candidate-set encoding before sequence assembly remains unchanged.
    \item \textbf{PIER + Flash + Pointwise} retains the complete FlashEvaluator architecture but replaces the listwise softmax objective with pointwise binary cross-entropy.
    \item \textbf{PIER + Flash + Listwise} denotes the full model with cross-sequence interaction and listwise softmax training.
\end{itemize}

Table~\ref{tab:ablation_modules} shows that removing cross-sequence interaction or replacing listwise supervision with pointwise BCE degrades performance. The first comparison supports the contribution of candidate-relative representation learning, while the second indicates that relative supervision is beneficial in this setting. These results suggest that the architectural and objective-level components provide complementary gains. For complementary diagnosis on RecFlow, we report the same four recommendation metrics as in the main experiment for consistency.

\begin{table}[t]
\centering
\caption{Ablation study of FlashEvaluator on RecFlow.}
\label{tab:ablation_modules}
\setlength{\tabcolsep}{3.5pt}
\begin{tabular}{lcccc}
\toprule
Model & N@6 & P@6 & R@6 & F1@6 \\
\midrule
PIER + Flash $-$ CrossSeq & 0.1919 & 0.0936 & 0.2438 & 0.1280 \\
PIER + Flash + Pointwise & 0.1912 & 0.0933 & 0.2427 & 0.1275 \\
PIER + Flash + Listwise & \textbf{0.1925} & \textbf{0.0938} & \textbf{0.2446} & \textbf{0.1284} \\
\bottomrule
\end{tabular}
\end{table}

\subsection{Effect of Candidate-Sequence Expansion}
\label{appdix:candidate_expansion}

To examine whether the efficiency advantage of FlashEvaluator can be converted into better candidate coverage, we vary the number of generated candidate sequences $K$ on RecFlow while keeping the generator, candidate pool, and other model configurations unchanged. As shown in Table~\ref{tab:candidate_scaling}, increasing $K$ consistently improves all evaluation metrics, indicating that FlashEvaluator can exploit a larger candidate search space to identify higher-quality sequences.

\begin{table}[t]
\centering
\caption{Effect of increasing the number of candidate sequences $K$ on RecFlow.}
\label{tab:candidate_scaling}
\setlength{\tabcolsep}{4.2pt}
\begin{tabular}{lccccc}
\toprule
Method & $K$ & N@6 & P@6 & R@6 & F1@6 \\
\midrule
PIER & 10 & 0.1910 & 0.0935 & 0.2431 & 0.1277 \\
PIER+Flash & 10 & 0.1925 & 0.0938 & 0.2446 & 0.1284 \\
PIER+Flash & 20 & 0.1952 & 0.0952 & 0.2483 & 0.1300 \\
PIER+Flash & 30 & \textbf{0.1978} & \textbf{0.0963} & \textbf{0.2506} & \textbf{0.1313} \\
\bottomrule
\end{tabular}
\end{table}

Compared with $K=10$, increasing the number of candidate sequences to $K=30$ improves NDCG@6 from $0.1925$ to $0.1978$. Precision, recall, and F1 also exhibit consistent improvements, suggesting that the enlarged search space increases both the ranking quality and the number of relevant items contained in the selected sequence. The gains gradually diminish as $K$ increases, which is expected because additional generated sequences increasingly overlap with previously explored candidates.
\section{Extended Related Work}
\label{appdix:extended_related_work}

\subsection{Relative Supervision and Set-Conditioned Prediction}

Inter-candidate dependence can be introduced into a ranking system at three distinct levels: the training objective, the decision rule, and the model architecture. These mechanisms should be distinguished, because coupling candidate scores in a loss does not necessarily imply that candidate representations interact during the forward pass.

\textbf{Relative training objectives.}
Pairwise and listwise learning-to-rank objectives provide supervision based on relative preferences among ranked objects. In RecSys and IR, BPR~\cite{rendle2009bpr} contrasts positive and negative items, whereas ListNet~\cite{cao2007listnet} defines a listwise objective over the scores of all items or documents in a ranking instance. In LLM alignment, DPO~\cite{rafailov2023dpo} and related methods such as SimPO~\cite{meng2024simpo} and ORPO~\cite{hong2024orpo} optimize
preferred--dispreferred response pairs. GRPO~\cite{shao2024grpo} instead samples a group of responses and constructs relative advantages by normalizing their rewards within the group. These objectives couple learning signals across candidates, but they primarily optimize scalar ranking scores or generation probabilities; they do not by themselves require representation-level interaction among candidates at inference time.

\textbf{Set-dependent decision rules.}
Some methods incorporate candidate relationships after candidate generation or scoring. Minimum Bayes risk decoding and consensus-based selection evaluate each hypothesis according to its expected utility or agreement with other hypotheses
~\cite{kumar2004mbr,rosti2007consensus,hildebrand2008combination}. Self-consistency similarly aggregates the answers obtained from multiple sampled reasoning paths~\cite{wang2023selfconsistency}. These methods are explicitly set-dependent, but the dependency is typically introduced by an external utility function, pairwise agreement computation, or voting rule rather than by learned interaction among candidate hidden representations.

\textbf{Set-conditioned model architectures.}
Architecture-level interaction has also been studied in RecSys and IR, primarily for ranking atomic items or documents. PRM
~\cite{pei2019personalized} models mutual influence among items in an input slate, while SetRank~\cite{pang2020setrank} jointly scores a set of documents with permutation-invariant modeling. IFA~\cite{yu2024ifa} processes all target items together, uses efficient cross-attention between the candidate set and the user behavior sequence, and additionally models relationships among target items. HoMer~\cite{chen2025homer} constructs
one set-wise sample for all items associated with a request, shares request-level computation, and performs cross-item interaction before predicting their CTRs. In LLM-based IR, setwise prompting presents multiple documents to the model within one comparison and reduces the number of model invocations required for ranking~\cite{zhuang2024setwise}.

These approaches establish that relative supervision and set-conditioned prediction are not unique to FlashEvaluator. Their evaluation granularity, however, is typically an atomic item or document. In FlashEvaluator, each candidate is itself a complete ordered sequence $L_k=(i_{k,1},\ldots,i_{k,l})$. The model therefore separates intra-sequence aggregation from interaction among the resulting sequence representations. Moreover, its architectural contribution is independent of the choice between pointwise, pairwise, and listwise training objectives: either a candidate-factorized or a joint evaluator
may, in principle, be trained with a relative loss. Our focus is the combination of complete-sequence set conditioning and request-local computation reuse rather than the first use of relative supervision.

\subsection{Efficient Attention and Multi-Candidate Execution}

Efficiency improvements for multi-candidate evaluation operate at different levels, including attention primitives, ranking-model architectures, and execution-level computation reuse. Separating these levels is important because a faster attention operator does not necessarily eliminate repeated computation across candidate inputs.

\textbf{Efficient attention primitives.}
Linear Attention~\cite{katharopoulos2020linearattention} reformulates attention to reduce its dependence on sequence length under a kernelized attention formulation. FlashAttention~\cite{dao2022flashattention} retains exact softmax attention while
reducing memory traffic through an IO-aware tiled implementation. Multi-head Latent Attention (MLA)~\cite{deepseek2024v2} compresses key/value states into a latent representation to reduce the memory footprint of autoregressive inference. These methods optimize the attention operator or its temporal KV cache and are complementary to reuse across multiple candidate sequences within the same request.

\textbf{Set-wise and hardware-aware ranking architectures.}
At the ranking-model level, IFA~\cite{yu2024ifa} jointly processes all target items and uses a Linear Transformer to efficiently model their interaction with a long user behavior sequence. HoMer~\cite{chen2025homer} performs set-wise CTR prediction for all items in a request within a single model invocation, thereby sharing item-independent features and computation. Setwise LLM ranking~\cite{zhuang2024setwise} reduces model calls and prompt-token consumption by comparing multiple documents in each invocation. RankMixer~\cite{zhu2025rankmixer} replaces quadratic self-attention with hardware-aware token mixing to improve model FLOPs utilization for feature interaction. These methods improve atomic-candidate
ranking or the efficiency of the underlying interaction operator, but they do not directly address repeated elements across multiple complete candidate sequences.

\section{Limitations and Future Directions}

The benefits of FlashEvaluator depend on the amount of element reuse, the number and diversity of candidate sequences, and the relative cost of shared encoding and sequence-level interaction. Dense cross-sequence attention may become non-negligible for very large $K$, while the public recommendation benchmarks rely on constructed candidate pools and pseudo-exposure sequences rather than complete industrial pipeline logs. Moreover, the theoretical results are conditional capacity and local-shift comparisons and do not automatically imply lower total risk in every setting. Future work may explore adaptive candidate-budget allocation, sparse or linear cross-sequence interaction, and applications to broader multi-candidate selection tasks such as LLM post-training, search, and question answering.

\end{document}